\documentclass[letter,11pt]{article}
\usepackage{amsmath} 
\usepackage{amssymb}  
\usepackage{amsthm}
\usepackage{graphicx}
\usepackage{times}
\newcommand{\Order}{\mathrm{O}}

\newcommand{\defeq}{\stackrel{\mbox{\scriptsize{\normalfont\rmfamily def.}}}{=}}

\renewcommand{\Vec}[1]{\mbox{\boldmath $#1$}}

\renewcommand\Pr{\mathop{\mathbf{Pr}}}
\DeclareMathOperator{\E}{\mathbf{E}}
\def\defeq{\mathrel{\mathop:}=}

\usepackage[capitalize]{cleveref}
\usepackage{enumitem}
\usepackage[ruled,linesnumbered,lined,commentsnumbered]{algorithm2e}

\newcommand{\Tbal}{T_\mathrm{bal}}

\newcommand{\aL}{\Gamma}
\newcommand{\tL}{K}
\newcommand{\Pos}{\mathsf{P}}
\newcommand{\Hei}{\mathsf{H}}
\newcommand{\ave}{\mu}
\newcommand{\rave}{\lceil \mu \rfloor}

\newcommand{\bL}{\overline{\Gamma}}
\newcommand{\btL}{\overline{K}}
\newcommand{\bPos}{\overline{\mathsf{P}}}
\newcommand{\bHei}{\overline{\mathsf{H}}}
\newcommand{\bave}{\overline{\mu}}
\newcommand{\brave}{\lceil \overline{\mu} \rfloor}

\renewcommand{\deg}{\mathrm{d}}
\newcommand{\degU}{D}
\newcommand{\rate}{r}
\newcommand{\fair}{F}
\newcommand{\RG}{\theta}
\newcommand{\bmG}{\boldsymbol{G}}
\newcommand{\Neb}{\mathsf{N}}

\newtheorem{theorem}{Theorem}[section]
\newtheorem{lemma}[theorem]{Lemma}

\newtheorem{proposition}[theorem]{Proposition}

\newtheorem{observation}[theorem]{Observation}

\title{
 An Analysis of Load-Balancing Algorithms\\
 on Edge-Markovian Evolving Graphs
}
\author{
  Takeharu Shiraga\footnote{Chuo University, Tokyo, Japan; e-mail: shiraga.076@g.chuo-u.ac.jp} \and
  Shuji Kijima\footnote{Shiga University, Hikone, Japan; e-mail: shuji-kijima@biwako.shiga-u.ac.jp}
}

\begin{document}
\maketitle
\begin{abstract}
Analysis of algorithms on time-varying networks (often called {\em evolving graphs})
  is a modern challenge in theoretical computer science.
 The \emph{edge-Markovian} is a relatively simple and comprehensive model of evolving graphs: 
  every pair of vertices which is not a current edge 
    independently becomes an edge with probability $p$  at each time-step, as well as
  every edge disappears with probability $q$. 
 Clearly, 
  the edge-Markovian graph changes its shape {\em depending} on the current shape, and 
   the dependency refuses some useful techniques for an independent sequence of random graphs 
    which often behaves similarly to a static random graph. 
 It motivates this paper 
   to develop a new technique for analysis of algorithms on edge-Markovian evolving graphs. 

 Specifically speaking, 
   this paper is concerned with \emph{load-balancing}, 
  which is a popular subject in distributed computing, and   
   we analyze the so-called \emph{random matching algorithms}, 
   which is a standard scheme for load-balancing. 
 We prove that major random matching algorithms 
   achieve nearly optimal load balance 
   in  $\Order(\rate \log (\Delta n))$ steps on edge-Markovian evolving graphs, 
   where $\rate \defeq \max\{p/(1-q), (1-q)/p\}$, 
     $n$ is the number of vertices (i.e., processors) and 
     $\Delta$ denotes the initial gap of loads unbalance. 
 We remark that the independent sequences of random graphs correspond to $\rate=1$. 
 To avoid the difficulty of an analysis 
   caused by a complex correlation with the history of an execution,  
  we develop a simple proof technique based on {\em history-independent} bounds. 
 As far as we know, this is the first theoretical analysis of
  load-balancing on randomly evolving graphs, 
  not only for the edge-Markovian but also for the independent sequences of random graphs. 


\

\end{abstract}

\section{Introduction}
In the real world, 
  connections or relationships between individuals continue to change time by time, 
  e.g., social relationships, peer-to-peer networks, wireless devices, etc. 
Such situations are naturally modeled by a graph changing its shape over time, called \emph{dynamic} graph. 
Analysis of algorithms on dynamic graphs, including both adversarial or stochastic changes, 
 is a modern challenge in theoretical computer science, and 
 it is widely studied in this decade~\cite{KO11,Augustine16,MS18}.

 One of the simplest models of dynamic graphs is 
  the \emph{dynamic Erd\H{o}s-R{\'e}nyi random graph}: 
  it is a time-series of Erd\H{o}s-R{\'e}nyi random graphs 
   $G_0, G_1, G_2, \ldots$, 
   where the random graphs are mutually \emph{independent}. 
 Theoretical analyses of 
    processes related to the spreading of infection or information 
  on the dynamic Erd\H{o}s-R{\'e}nyi random graph 
 have been studied 
  to investigate the relationship between the connectivity threshold $p$ and the speed of the spreading, 
 for instance, 
 SIR (susceptible-infected-removed) model \cite{DGM08},
 random walks \cite{AKL18}, 
 and radio broadcasting~\cite{EG06,CMPS09}. 

 The dynamic Erd\H{o}s-R{\'e}nyi random graph is simple enough to mathematically analyze, 
   while it might be a strong assumption for a model of real networks that 
   the graph changes its shape in the next time step completely different from the current one. 
 Clementi et al.~\cite{CMMPS10} introduced 
  the \emph{edge-Markovian evolving graph} as a more general model of dynamic graphs, 
   which includes the dynamic Erd\H{o}s-R{\'e}nyi random graph. 
 Precisely, 
   it is a random sequence of graphs $G_0, G_1, G_2, \ldots$ with the same vertex set $V$, 
   and in an update from $G_t$ to $G_{t+1}$, 
  each (existing) edge $e\in E(G_t)$ independently disappears with probability $q$, 
  and each (not-existing) edge $e\in \binom{V}{2}\setminus E(G_t)$ independently appears with probability $p$.
Note that it is identical to the dynamic Erd\H{o}s-R{\'e}nyi random graph when $q = 1-p$.
Recently, there have been many works on the model concerning
fundamental processes, e.g., 
flooding~\cite{CMMPS10,BCF11,ClementiDC15,CMPS11}, 
rumor spreading~\cite{CCDFPS16}, 
and random walk~\cite{LMS18,CSZ20}.

In this paper, we are concerned with the \emph{load-balancing problem}.
Suppose that each vertex $v$ initially holds $L(v)\in \mathbb{N}$ tokens.
We aim to reallocate the tokens as equally as possible under the assumption that each vertex is only allowed to exchange tokens with its neighbors.
The main interest of the study is the number of time steps required to reach the almost balanced configuration.
There are many studies concerned with the load-balancing paradigm.
The load-balancing problem naturally models the coordination of tasks in distributed processor networks and parallel machines \cite{Cybenko89}.
This problem is also referred to as the distributed averaging problem, which arises in many applications, e.g., coordination of autonomous agents, estimation, and data fusion, on distributed networks such as sensors, wireless ad-hoc, and peer-to-peer networks~\cite{BGPS06}.
In computational physics, load-balancing algorithms appear to simulate large and complicated correlation systems such as molecular dynamics~\cite{BBK91} and electrostatic plasma~\cite{FLD93}.

The load-balancing problem has been well studied on \emph{static} graphs.
Particularly, there are many theoretical studies for a type of algorithms called \emph{random matching algorithms}~\cite{GM96,BGPS06,FS09,SS12,CS17}. 
In a random matching algorithm, at each discrete time step, we generate a random matching $M\subseteq E$ with some property.
Then, for each matching edge $\{v,u\}\in M$, we reallocate tokens on $v$ and $u$ by the \emph{random rounding}: $(L(v), L(u))$ to $\bigl(\bigl\lceil \frac{L(v)+L(u)}{2}\bigr \rceil,\bigl\lfloor \frac{L(v)+L(u)}{2}\bigr \rfloor\bigr)$ or $\bigl(\bigl\lfloor \frac{L(v)+L(u)}{2}\bigr \rfloor,\bigl\lceil \frac{L(v)+L(u)}{2}\bigr \rceil\bigr)$ with probability $1/2$ each.
For example, the \emph{LR algorithm}~\cite{GM96}
is known as a specific algorithm to generate a random matching locally.
For such algorithm, several works~\cite{FS09,SS12,CS17} studied the \emph{discrepancy}, which is the maximum difference of tokens $\max_{v\in V}L(v)-\min_{v\in V}L(v)$.
For example, Sauerwald and Sun~\cite{SS12} showed that a random matching algorithm reaches the configuration with constant discrepancy on any connected regular graph.
Formally, let $\Delta$ be the initial discrepancy and $\lambda$ be the second largest eigenvalue of the diffusion matrix\footnote{The diffusion matrix $P$ is defined by $P(v,u)=1/(2d_{\max}(G))$ if $\{v,u\}\in E(G)$, $P(v,v)=1-\deg(G,v)/(2d_{\max}(G))$, and $P(v,u)=0$ otherwise.}.
They showed that, for any connected regular graph, the discrepancy is at most some constant w.h.p.~after $O\bigl(\frac{\log (\Delta n)}{1-\lambda}\bigr)$ steps.

\subsection{Results}
We study the performance of random matching algorithms on edge-Markovian evolving graphs, although all previous works are concerned with static graphs, as far as we know. 
Let $\aL_t=(\aL_t(v))_{v\in V}\in \mathbb{N}^V$ denote the configuration of tokens at time $t\geq 0$ and let $\lceil \cdot \rfloor$ denote the rounding operator to the nearest integer\footnote{$\lceil x \rfloor\defeq \lceil x-1/2\rceil$ for $x\in \mathbb{R}$}.
Let $\Delta\defeq \max_{v\in V}\aL_0(v)-\min_{v\in V}\aL_0(v)$ be the initial discrepancy.
We study the following \emph{balancing time} $\Tbal=\Tbal(\aL_0)$ as a measure of balancing:
\begin{align}
\Tbal\defeq \min\bigl\{t\geq 0: \aL_t(v)\in \{\rave-1, \rave, \rave+1\}\textrm{ for all $v\in V$}\bigr\}. 
\end{align}
%
We show the following theorem for random matching algorithms including the LR algorithm~\cite{GM96}. 
The formal condition required to random matching algorithms is in \cref{sec:definitions}.
\begin{theorem}
\label{thm:RM_main}
Consider a random matching algorithm on an edge-Markovian evolving graph of $\max\{p,1-q\}=\Omega(1/n)$.
For any initial configuration, $\Tbal=\Order\left(\rate \log (\Delta n)\right)$ w.h.p.
where $\rate \defeq \frac{\max\{p,1-q\}}{\min\{p,1-q\}} = \max\left\{\frac{1-q}{p},\frac{p}{1-q}\right\}$.
\end{theorem}

\Cref{thm:RM_main} gives a simple upper bound of the balancing time for a wide range of parameters $p,q$ of edge-Markovian evolving graphs.
%
A simple setting for $r$ to be constant is the case that both $p$ and $q$ are constants. 
%
Another condition for $r$ to be a constant is the case of $p\sim 1-q$.
This case includes dynamic Erd\H{o}s-R{\'e}nyi random graphs (the case of $p=1-q$). 
In this situation, even if $p=(1-\varepsilon)/n$ for a constant $0<\varepsilon<1$, we can apply \cref{thm:RM_main} and obtain $\Tbal=O(\log(\Delta n))$.
Although each $G_t$ does not contain any giant component w.h.p.~in this case, it is identical with the known upper bound of the complete graph~\cite{SS12}. 

Berenbrink et al.~\cite{BFKK19} investigated the balancing time for the \emph{simple load-balancing algorithm}.
In this algorithm, at each time step, an edge $\{v,u\}$ is randomly picked and tokens on $v$ and $u$ are reallocated by the random rounding (see \cref{sec:simple} for the formal definition). They showed that, on the complete graph $K_n$, $\Tbal=O(n\log(\Delta n))$ w.h.p.~for this algorithm. 
In this paper, we also give the following result generalizing it.
\begin{theorem}
\label{thm:SLB_main}
Consider the simple load-balancing algorithm on an edge-Markovian evolving graph of $\max\{p,1-q\}=\Omega(1/n)$.
Let $\rate \defeq \frac{\max\{p,1-q\}}{\min\{p,1-q\}}$.
Then, for any initial configuration, $\Tbal=\Order\left(\rate n\log (\Delta n)\right)$ w.h.p.
\end{theorem}
%

%

Our analysis is quite simple,
while analyses of load-balancing algorithms or dynamic graphs, including the edge-Markovian evolving graph, often require advanced mathematics about transition matrices~\cite{FS09,SS12,Saloff-Coste09,Saloff-Coste11,CSZ20}.
Our proof technique is based on token-based analysis for static complete graphs~\cite{BFKK19}: 
Suppose $\tL=\sum_{v\in V}\aL_t(v)$ tokens have distinct labels and each stacked token on a vertex $v\in V$ is allocated a height in $1,2,\ldots, \aL_t(v)$. 
In \cite{BFKK19}, the authors proposed a movement rule of tokens (called the skip mode) corresponding to an update of a configuration, where every token's height does not increase. 
Furthermore, they guarantee that any token's height sufficiently decreases w.h.p.~on $K_n$.
%
We deal with this technique more carefully to apply it to the edge-Markovian evolving graph. 
In particular, we take care of the imaginary tokens called \emph{inverted tokens} to discuss the minimum height $\min_{v\in V}\aL_t(v)$ and $\max_{v\in V}\aL_t(v)$ together (\cref{sec:analytic_framework}).
It enables us to provide a framework for analyzing the balancing time only using token-based analysis (\cref{thm:RM_general}).
Our main theorem is derived from the framework and a careful estimation of conditional probabilities concerning a random matching on the edge-Markovian evolving graph (\cref{lem:prob_RM}).

\subsection{Related works}
Several early works~\cite{Cybenko89,SS94,GM96,MGS98} consider the load-balancing problem with \emph{continuous} load ($L(v)\in \mathbb{R}$). In other words, each node $v\in V$ does not need any rounding but can exchange the ideal amount of load with its neighbors, e.g., $L(v)/2\in \mathbb{R}$. 
In this setting, the propagation of the load on a graph is highly related to the probability distribution of Markov chains. 
For example, the time taken for some balancing models to reach a constant discrepancy have been shown using the second largest eigenvalue or the graph conductance \cite{Cybenko89,SS94}.
Compared to the continuous case, the load-balancing problem with discrete tokens ($L(v)\in \mathbb{N}$) is much harder to analyze due to the rounding errors caused in each step and each vertex.
Throughout the paper, we consider the load-balancing problem with discrete tokens case.

\emph{Diffusion-based algorithms} have been also well studied for the load-balancing problem
\cite{RSW98,FGS12,AB13,BCFFS15,BKKMU19}.
Roughly speaking, in a diffusion-based algorithm on a $d$-regular graph, every vertex $v\in V$ sends $\lfloor L(v)/d\rfloor$ tokens to its all neighbors at each time step.
There are many works on the discrepancy of diffusion-based algorithms. 
For example, Rabani et al.~\cite{RSW98} showed that, on any $d$-regular graphs, the discrepancy of a diffusion-based algorithm using the rounding down is at most $O\bigl(\frac{d\log n}{1-\lambda}\bigr)$ after $O\bigl(\frac{\log(\Delta n)}{1-\lambda}\bigr)$ steps.
To obtain a smaller discrepancy, diffusion-based algorithms combining the rounding up and down \cite{FGS12}, 
distributing tokens by the round-robin algorithm (called the rotor-router) \cite{AB13}, and using a randomized rounding \cite{BCFFS15} have been also studied.
Recently, Berenbrink et al.~proposed a sophisticated deterministic rounding framework and showed the $O(d)$-discrepancy~\cite{BKKMU19}.

Random matching algorithms are originally introduced by Ghosh and Muthukrishnan~\cite{GM96}, with a motivation of a more efficient way to send tokens than the diffusive way. 
They proposed an algorithm referred to as the LR algorithm, that generates random matching in a distributed way. 
Note that the LR algorithm uses the degree information on the adjacent vertices if the graph is irregular.
Boyd et al.~\cite{BGPS06} proposed an algorithm generating random matching called the distributed synchronous algorithm, which uses the maximum degree information.
The discrepancy of these algorithms has been studied in \cite{FS09,SS12,CS17}, e.g., Friedrich and Sauerwald~\cite{FS09} showed that the discrepancy after $O\bigl(\frac{\log (\Delta n)}{1-\lambda}\bigr)$ steps is at most $O\bigl(\sqrt{\frac{\log^3n}{1-\lambda}}\bigr)$ w.h.p.~on any $d$-regular graph.
As mentioned above, a constant discrepancy on any $d$-regular graph has been shown in~\cite{SS12}.
Several works focus on deterministic (periodic) matching algorithms, called balancing circuit models~\cite{RSW98,FS09,SS12}.

%
The simple load-balancing introduced in \cite{BFKK19} appears as a subroutine in the population protocol \cite{BKR19,BEFKKR20}.
Recently, Huang and Wang~\cite{HW22} study the balancing time of the simple load-balancing on complete bipartite graphs.
\section{Preliminaries}

\subsection{Edge-Markovian graph, and other terminologies about (static) graphs}\label{sec:edge-Markovian}
 An {\em edge-Markovian graph} is a sequence of (static) graphs $\bmG=G_0,G_1,G_2\ldots$ 
  where every graph $G_t=(V,E_t)$ ($t=0,1,2,\ldots$) is undirected and simple. 
 An edge-Markovian graph $\bmG$ is characterized 
   by $G_0=(V,E_0)$, $p \in (0,1]$ and $q \in [0,1)$. 
 The vertex set $V$ is invariant with respect to $t$, where let $n=|V|$ for convenience. 
 The edge set $E_t$ ($t=1,2,\ldots$) is a random variable depending only on $E_{t-1}$:
   when a distinct vertex pair $\{u,v\}$ is NOT an edge of $E_{t-1}$, the pair $\{u,v\}$ becomes an edge of $E_t$ with probability $p$; 
   when $\{u,v\}$ is an edge of $E_{t-1}$, the pair $\{u,v\}$ withdraws from $E_t$ with probability $q$. 
 In other words, 
   let 
\begin{align*}
  X_t(\{u,v\}) = \begin{cases} 
  0 & (\mbox{if $\{u,v\} \not\in E_t$})\\
  1 & (\mbox{if $\{u,v\}       \in E_t$})
\end{cases}
\end{align*}
  for any distinct pair of vertices $\{u,v\} \in \binom{V}{2}$ and $t=0,1,2,\ldots$, 
 and 
   let $P=(p_{ij}) \in \mathbb{R}^{2 \times 2}$ be a probability matrix given by 
\begin{align*}
  p_{ij} = \Pr[X_{t+1}(\{u,v\}) = j-1 \mid X_t(\{u,v\}) =i-1]
\end{align*}
 for $i \in \{1,2\}$ and $j \in \{1,2\}$, 
then 
\begin{align}
 P = \begin{pmatrix} 1-p & p \\ q & 1-q\end{pmatrix}
\label{eq:Edge-Markov-prob}
\end{align}
holds (see also \cref{fig:Edge-Markov-prob}). 
 It might be worth to mention the fact, though we do not use it in this paper, that 
  $\Pr[\{u,v\} \in E_t]$ approaches\footnote{ 
   We emphasis that $p+q \neq0$ since $p \in (0,1]$. 
   The condition $p+q \neq0$ excludes the case of $\bmG$ being static, i.e., $G_0=G_1=G_2=\cdots$ hold if and only if $p=q=0$. 
  } $p/(p+q)$ asymptotic to $t$. 
 Furthermore, 
  if $p+q=1$ then $\Pr[\{u,v\} \in E_t]=p/(p+q)$ always hold for $t=1,2,\ldots$, 
  otherwise   $\Pr[\{u,v\} \in E_t]$ is enough close to $p/(p+q)$ 
   in $\Order(1/\log(|1-p-q|))$ time steps. 
  See \cref{sec:conv} for more detail. 
\begin{figure}
\centering
\begin{tabular}{|c|c|c|}
\hline
  & $\{u,v\} \not\in E_t$ & $\{u,v\} \in E_t$ \\
\hline
  $\{u,v\} \not\in E_{t-1}$ & $1-p$ & $   p$ \\
\hline
  $\{u,v\}       \in E_{t-1}$ & $   q$ & $1-q$ \\
\hline
\end{tabular}
\caption{Table of conditional probabilities. \label{fig:Edge-Markov-prob}}
\end{figure}

 Since this paper is concerned with an arbitrary initial graph $G_0$ as an worst case analysis, 
  we in this paper let $\bmG(n,p,q)=G_0,G_1,G_2,\ldots$ represent an edge-Markovian graph, 
   i.e., $G_0$ in the characterization of an edge-Markovian graph is replaced by just the number of vertices $n$.

\paragraph{Other terminologies about static graphs}\label{sec:definitions}
 Let $G=(V,E)$ be a (static) undirected simple graph with $n=|V|$ vertices. 
 Let $\mathcal{G}$ denote the entire set of graphs with $n$ vertices. 
 Let $\deg_G(v)= |\{\{v,u\}\in E \mid u\in V \}|$ denote the {\em degree} of a vertex $v\in V$. 
 A set of edges $M \subseteq E$ is a {\em matching} 
  if every pair of edges never shares the end vertices, 
  i.e., $e=\{u,v\}$ and $f=\{w,x\}$ satisfies $e \cap f = \emptyset$ for any distinct $e,f \in M$. 
 For convenience, 
  let $M(v)$ denote the vertex matched with a vertex $v \in V$ in a matching $M$, 
  i.e.,  $u=M(v)$ if $\{v,u\} \in M$.

\subsection{Random matching algorithm on edge-Markovian graph}\label{sec:algorithm}
 Random matching is a comprehensive method for load balancing on graphs. 
 This section 
   describes the {\em random matching algorithm on an edge-Markovian graph}, and 
   describes the main theorem of the paper. 

\subsubsection{Algorithm description}
 Let $\bmG(n,p,q)=G_0,G_1,G_2\ldots$ be an edge-Markovian graph, 
    where $G_t = (V,E_t)$ denotes the graph at time $t$. 
 Let $\aL_t \in \mathbb{N}^V$ denote the configuration of tokens at time $t = 0,1,2,\ldots$, 
  where $\aL_t(v)$ denotes the number of token on $v \in V$. 
 The initial configuration $\aL_0$ is given arbitrarily. 
 For convenience, 
  let $\tL \defeq \sum_{v \in V} \aL_0(v)$ denote the total number of tokens, which is an invariant of $t$, and 
  let  
\begin{align}
 \Delta\defeq \max_{v\in V}\aL_0(v)-\min_{v\in V}\aL_0(v).  
\label{def:Delta}
\end{align} 

 The {\em random matching algorithm} 
   stochastically updates the token configuration $\aL_t \mapsto \aL_{t+1}$ as follows. 
 At time $t$, 
   the algorithm randomly chooses a random matching $M_t \subseteq E_t$ 
    according to some probability distribution $\mathcal{D}_t$ (see \cref{sec:applications} for examples). 
 For every matching edge $\{v,u\}\in M_t$, 
 we choose either 
\begin{align}
(\aL_{t+1}(v),\aL_{t+1}(u)) = \begin{cases}
 \left(\left \lceil\frac{\aL_{t}(v)+\aL_{t}(u)}{2} \right \rceil, \left \lfloor\frac{\aL_{t}(v)+\aL_{t}(u)}{2} \right \rfloor \right) 
   &\hspace{2em} \mbox{(i),\quad or } \\[2ex]
 \left(\left \lfloor\frac{\aL_{t}(v)+\aL_{t}(u)}{2} \right \rfloor, \left \lceil\frac{\aL_{t}(v)+\aL_{t}(u)}{2} \right \rceil \right) 
   &\hspace{2em} \mbox{(ii)}
\end{cases} 
\label{def:simplebalance}
\end{align}
  with probability $1/2$.\footnote{
   This probability $1/2$ is just for simplicity of the algorithm description (by symmetry of $v$ and $u$), and it is not essential. 
   We can replace the probability $1/2$ with any other probability, and 
    it is easy to apply the argument of the paper to the variant. 
   }
 For all unmatched vertices $w\in V$, i.e., $\{w,w'\}\not\in M_t$ for any $w'\in V$, 
  set $\aL_{t+1}(w)=\aL_t(w)$.
\begin{algorithm}[t]
\caption{Random matching algorithm on $\bmG$}
\label{fig:algorithm}
\SetKwInOut{Input}{Input}\SetKwInOut{Output}{Output}
\Input{$G_0=(V,E_0)$, $p \in (0,1]$, $q \in [0,1)$ and $ \aL_0 \in \mathbb{N}^V$}
\Output{$\aL_T$}
\BlankLine
\For{$t =0$ \KwTo $T-1$}{
{generate a random matching $M_t \subseteq E_t$}\tcp*[r]{see Algorithms \ref{alg:simple}, \ref{alg:LR} and \ref{alg:ds}}
 \For{$\{u,v\} \in M_t$}{
  {\cref{def:simplebalance}}\;
 }
 \tcp{we obtain $\aL_{t+1}$} 
{($G_{t+1}$ is generated from $G_t$)}\tcp*[r]{see \cref{sec:edge-Markovian}}
}
\KwRet{$\aL_T$}\;
\end{algorithm}

  A very simple example of the random matching algorithm\footnote{
   ``Choose a matching randomly'' is a profound problem in contrast to it looks: 
      for instance, choose a matching (almost) uniformly at random in bipartite graph 
      had been investigated for a long time, see e.g., \cite{JS89,JSV04}. 
     Of course, it is easy for some certain distributions. 
  } is 
    the {\em simple load balancing}, which is analyzed 
      for static complete graphs by \cite{BFKK19} and for complete bipartite graph by \cite{HW22}: 
  the algorithm chooses a vertex $v \in V$ uniformly at random, 
   then chooses a single edge incident to $v$ as far as exist, 
   i.e.,  the random matching consists of just a single edge as $M_t=\{\{u,v\}\}$, or empty. 
 An edge $\{u,v\} \in E_t$ is 
   chosen with probability $\frac{1}{n} ( \frac{1}{\deg_{G_t}(u)}+\frac{1}{\deg_{G_t}(u)})$. 
 See \cref{sec:applications} for other examples. 

 An execution of a random matching on a edge-Markovian graph 
  is represented by a sequence of a triplet $(\aL_t,G_t,M_t)$. 
  For convenience, 
   let 
   $\mathcal{E}_t = (\aL_0,G_0,M_0), (\aL_1,G_1,M_1), \ldots, (\aL_{t-1},G_{t-1},M_{t-1}), \aL_t$
   denote the history of an execution until time $t$.  
 In the execution, we remark, 
   $G_t$ depends only on $G_{t-1}$, and 
   $M_t$ is chosen according to $\mathcal{D}_t$ 
     depending on $G_t$ (and possibly depends on $\mathcal{E}_t$, too). 
  Depending on $\aL_t$ and $M_t$, 
   the configuration $\aL_{t+1}$ is 
   probabilistically determined (recall \cref{def:simplebalance}).

\subsubsection{$\fair$-fair condition and main theorem}
 We say the distribution $\mathcal{D}_t$, which a random matching $M_t \subseteq E_t$ follows, 
  satisfies the {\em $\fair$-fair} condition (for $\mathcal{E}$) 
 if there exists $\fair > 0$ such that 
\begin{align}
    \Pr\left[\{v,u\} \in M_t \,\middle|\, G_t =G,\,\mathcal{E}\right] 
    \geq \frac{\fair}{\max\{\deg_G(v),\deg_G(u)\}} 
    \label{def:matching}
\end{align}
  holds for any $\{u,v\} \in E_t$, for any graph $G\in \mathcal{G}$  and for any other events $\mathcal{E}$, 
   where we assume the execution $\mathcal{E}_t$ as $\mathcal{E}$ but not limited to. 
 We also say a random matching algorithm satisfies the $\fair$-fair condition 
  if $M_t$ satisfies the condition \cref{def:matching} at any time $t$ in the algorithm. 
 We remark 
   that $\fair$ can be a function of $n$ such as $1/n$: 
  we will show two examples of specific random matching algorithms in \cref{sec:applications}, 
   where the algorithms respectively satisfy $\fair=1/8$ and $\fair=1/n$ conditions. 
 Notice that 
   $\fair$ cannot be more than $1$ as far as $M_t$ is a matching\footnote{
   Proof: 
   Suppose $v \in V$ satisfies $\deg_G(v) = \max_{u\in V}\deg_G(u)$. 
  Under \eqref{def:matching}, 
    the expected number of partners of $v$ in a random matching $M_t$ satisfies 
    $\E[|\{ \{u,v\} \in E \mid M_t(v)=u\}|] 
      = \sum_{\{u,v\} \in E}\Pr\left[\{u,v\} \in M_t \,\middle|\, G_t =G,\,\mathcal{E}\right]
      \geq \sum_{\{u,v\} \in E} \frac{\fair}{\deg_G(v)}
      =\deg_G(v)\frac{\fair}{\deg_G(v)} = \fair$. 
  Since any matching $M_t$ must satisfy $|\{ \{u,v\} \in E \mid M_t(v)=u\}| \leq 1$, 
   we get $\fair \leq 1$. 
  }. 

 Now, we are ready to describe our main theorem. 
\begin{theorem}
\label{thm:RM_general}
 Suppose a pair of $p \in (0,1]$ and $q\in[0,1)$ satisfy $\max\{p,1-q\}\geq \RG/n$ for a constant\footnote{
  The assumption of a {\em constant} $\theta$ is just for simplicity of the arguments.
  We can establish a similar theorem for $\theta = {\rm o}(n)$ 
   if we allow $c$ to be a function of $p$ and $q$.  
  See the supplementary argument of the theorem just below, and the definition of $c_*$. 
 } $\RG > 0$. 
 Let $\bmG(n,p,q)=G_0,G_1,G_2,\ldots$ be an edge Markovian graph, 
  and let $\aL_0 \in \mathbb{N}^V$ be an initial configuration of tokens. 
  If the random matching algorithm satisfies the $\fair$-fair condition ($0 < \fair \leq 1$) given in \eqref{def:matching}
   then its balancing time $\Tbal$ satisfies 
\begin{align}
\Pr \left[\Tbal \leq \frac{c \rate \log \left(\frac{\Delta n}{\varepsilon}\right)}{\fair}\right] \geq 1-\varepsilon 
\end{align}
  for any $\varepsilon$ ($0<\varepsilon<1/4$), 
 where $c$ is an appropriate constant, 
 $\rate \defeq \max\left\{\frac{p}{1-q}, \frac{1-q}{p}\right\} \geq 1$ and 
 $\Delta\defeq \max_{v\in V}\aL_0(v)-\min_{v\in V}\aL_0(v)$ (cf., \eqref{def:Delta}). 
\end{theorem}
 In fact, 
  we give a constant $c= 91 / c_* $ for \cref{thm:RM_general} in our proof, where 
\begin{align}
 c_* \defeq \dfrac{\left(1-\exp\left(-\frac{\RG}{3}\right)\right)^2}{2+\frac{1}{\RG}}.
\label{def:c*}
\end{align}
 This constant $c$ is not optimized at all due to the simplification of the arguments. 
 We remark that $c_*=c_*(\RG)$ is monotone increasing with respect to $\RG$ $(\RG > 0)$, 
  such that $c_*(+0)=0$, $c_*(1)\simeq 0.02678$ and $c_*(\infty)=0.5$. 

\section{A Lemma for Theorem \ref{thm:RM_general}}
\label{sec:randommatching}
 As a preliminary step of the proof of \cref{thm:RM_general}, 
  this section establishes a key lemma 
  about a probability of a random matching in an edge-Markovian graph. 
 A random matching $M_t \subseteq E_t$ clearly depends on $G_t$, 
   and the graph $G_t$ depends on $G_{t-1}$.  
 It makes the analysis of $\aL_t$ complex 
   compared with some simpler models of dynamic graphs, so called independent random graph model, 
    where all $G_t$ are mutually independent. 
 To avoid the difficulty caused by the history-dependence of the edge-Markovian model,  
  we give a useful lower bound of the probability 
    that a specific vertex $v$ is matched with a desired vertex of $U_t \subseteq V$, 
     where $U_t$ is given randomly depending on the history of an execution $\mathcal{E}_t$. 
 The lower bound plays a key role in the proof of \cref{thm:RM_general} in \cref{sec:proof}.  
\begin{lemma}
\label{lem:prob_RM}
 Suppose a pair of $p \in (0,1]$ and $q\in[0,1)$ satisfy $\max\{p,1-q\}\geq \RG/n$ for a constant\footnote{
  The assumption of a {\em constant} $\theta$ is just for $c_*$ to be a constant to $p$ and $q$. 
   We can establish a similar lemma for $\theta = {\rm o}(n)$ if we allow $c_*$ to be a function of $p$ and $q$. 
   See the definition of $c_*$ in \cref{def:c*}. 
 } $\RG > 0$. 
 Let $\bmG(n,p,q)=G_0,G_1,G_2,\ldots$ be an edge Markovian graph, 
  and let $M_t \subseteq E_t$ be a random matching of $G_t=(V,E_t)$ according to a distribution $\mathcal{D}_t$. 
 Let $U_t\subseteq V$ be any random variable. 
 Note that $M_t$ and $U_t$ may depend on each other as well as any other random variables; 
  for convenience let $\mathcal{E}$ denotes any possible event, 
  e.g., an execution $(\aL_0,G_0,M_0), (\aL_1,G_1,M_1), (\aL_2,G_2,M_2), \ldots$ in \cref{sec:algorithm}.  
If $\mathcal{D}_t$ satisfies the $\fair$-fair condition~\eqref{def:matching} for $\mathcal{E}$ then 
\begin{align}
    \Pr\left[M_t(v) \in U_t \,\middle|\,\mathcal{E}\right] 
    &\geq \frac{c_*\fair}{\rate}\E\left[\frac{|U_t|}{n} \,\middle|\, \mathcal{E}\right] 
\label{eq:RMkey-0}
\end{align}
holds, where 
 $c_*$ is a constant given by \cref{def:c*}, and  
 $\rate\defeq \max\left\{\frac{p}{1-q},\frac{1-q}{p}\right\}$.  
\end{lemma}
\begin{proof}
 For convenience, 
   let $\degU \defeq 2n\max\{p,1-q\}$, and let 
\begin{align*}
    \mathcal{G}_{v,u}
    &=\left\{G=(V,E)\in \mathcal{G} \,\middle|\,\begin{aligned}
         &\{v,u\}\in E \\
         &\deg_G(v)\leq \degU+1, \\
         &\deg_G(u)\leq \degU+1
    \end{aligned}\right\}
\end{align*}
  for $v,u\in V$. 
  Considering the marginal probabilities, we see that 
\begin{align}
    &\Pr\left[M_t(v)\in U_t\,\middle|\,\mathcal{E}\right] \nonumber\\
    &=\footnotemark
    \sum_{U \in 2^V}
         \Pr\left[M_t(v)\in U \,\middle|\, U_t=U,\mathcal{E}\right]
         \cdot \Pr\left[U_t = U,\,\middle|\, \mathcal{E}\right]
    \nonumber\\
    &=\sum_{U \in 2^V}\sum_{u\in U}
         \Pr\left[M_t(v)=u \,\middle|\, U_t=U,\mathcal{E}\right]
         \cdot \Pr\left[U_t=U,\,\middle|\, \mathcal{E}\right] 
    \nonumber\\
    &=\sum_{U \in 2^V}\sum_{u\in U}
         \Pr\left[\{u,v\} \in M_t \,\middle|\, U_t=U, \mathcal{E}\right]
         \cdot \Pr\left[U_t=U,\,\middle|\, \mathcal{E}\right] 
    \nonumber\\
    &=\sum_{U \in 2^V}\sum_{u\in U}\sum_{G\in \mathcal{G}_{v,u}}
         \Pr\left[ \{v,u\} \in M_t  \,\middle|\, U_t = U,G_t=G,\mathcal{E}\right] 
         \cdot\Pr\left[U_t = U,G_t=G\,\middle|\, \mathcal{E}\right]
    \nonumber\\
    &\geq
         \frac{\fair}{\degU+1}
         \sum_{U \in 2^V}\sum_{u\in U}\sum_{G\in \mathcal{G}_{v,u}}
         \Pr\left[U_t = U,G_t=G\,\middle|\, \mathcal{E} \right]
         \hspace{2em}(\mbox{by \cref{def:matching} and $\max\{\deg_G(u),\deg_G(v)\} \leq \degU$})
    \nonumber\\
    &=
         \frac{\fair}{\degU+1}
         \sum_{U \in 2^V}\sum_{u\in U}
        \Pr\left[U_t = U,G_t\in \mathcal{G}_{v,u}\,\middle|\, \mathcal{E} \right]
    \nonumber\\
    &=
         \frac{\fair}{\degU+1}
         \sum_{U \in 2^V}\sum_{u\in U}
    \Pr\left[G_t\in \mathcal{G}_{v,u}\,\middle|\, U_t = U,\mathcal{E}\right]
    \Pr\left[U_t = U\,\middle|\, \mathcal{E} \right]
    \label{eq:RMkey-1}
\end{align}
\footnotetext{
     We in this paper assume in a marginal probability $\sum_{x \in \Omega} \Pr[A \mid X=x]\Pr[X=x]$
        that $\Pr[A \mid X=x]\Pr[X=x]=0$ if $\Pr[X=x]=0$. 
     In other words, ``$x \in \Omega$'' in the subscription of $\sum$ 
       is an abbreviation of ``$x \in \{x' \in \Omega \mid \Pr[X=x']\neq 0\}$.'' 
    }
 hold. 

Concerning the term $\Pr\left[G_t\in \mathcal{G}_{v,u} \,\middle|\, U_t = U,\mathcal{E} \right]$ in \cref{eq:RMkey-1}, 
 we will claim 
 \begin{align}
\Pr\left[G_t\in \mathcal{G}_{v,u}\,\middle|\, U_t = U,\mathcal{E} \right] 
    \geq 
    \min\{p,1-q\}\left(1-\exp\left(-\frac{\RG}{3}\right)\right)^2
    \label{eq:RMkey-4}
\end{align}
 holds for any $U \in 2^V$ and $u \in U$. In fact, 
\begin{align}
&\Pr\left[G_t\in \mathcal{G}_{v,u}\,\middle|\, U_t = U,\mathcal{E} \right]\nonumber \\
&=\Pr\left[\begin{aligned}
     &\{v,u\}\in E_t,\\
     &\deg_{G_t}(v)\leq \degU+1,\\ 
     &\deg_{G_t}(u)\leq \degU+1
\end{aligned} \,\middle|\, U_t = U,\mathcal{E} \right]\nonumber \\
    &=
       \Pr\left[\{v,u\}\in E_t \,\middle|\, U_t = U,\mathcal{E} \right] 
    \cdotp  \Pr\left[\begin{aligned}
        &\deg_{G_t}(v)\leq \degU+1, \\
        &\deg_{G_t}(u)\leq \degU+1
    \end{aligned} \,\middle|\, 
         \{v,u\}\in E_t,\, U_t = U,\, \mathcal{E} 
    \right] \nonumber\\
    &\geq 
    \min\{p,1-q\}\Pr\left[\begin{aligned}
        &\deg_{G_t}(v)\leq \degU+1, \\
        &\deg_{G_t}(u)\leq \degU+1
    \end{aligned} \,\middle|\, 
         \{v,u\}\in E_t,\, U_t = U,\, \mathcal{E}
    \right]
    \label{eq:RMkey-2}
\end{align}
holds, 
where the last inequality follows 
$ \Pr\left[\{v,u\}\in E_t \,\middle|\, U_t = U,\mathcal{E}_t\right]\geq \min\{p,1-q\}$ 
 since $\{v,u\}\in E_t$ depends only on $E_{t-1}$ in the edge-Markovian model (cf.,~\cref{sec:edge-Markovian}) 
 and the probability is $1-q$ (if $\{v,u\}\in E_{t-1}$) or $p$ (if $\{v,u\}\not\in E_{t-1}$). 

To evaluate the second term of \cref{eq:RMkey-2}, 
  let $X_t(\{w,w'\})$ for $\{w,w'\} \in \binom{V}{2}$ be independent binary random variables 
  given by $X_t(\{w,w'\})=1$ if $\{w,w'\}\in E_t$, otherwise  $X_t(\{w,w'\})=0$. 
 Then, $\deg_{G_t}(w)=\sum_{w'\in V\setminus \{w\}}X_t(\{w,w'\})$ holds. 
Then, 
\begin{align}
\lefteqn{\Pr\left[\begin{aligned}
        &\deg_{G_t}(v)\leq \degU+1, \\
        &\deg_{G_t}(u)\leq \degU+1
    \end{aligned} \,\middle|\, 
         \{v,u\}\in E_t,\, U_t=U,\, \mathcal{E}
    \right]}\nonumber\\
    &= \Pr\left[\begin{aligned}
         &\textstyle\sum_{w\in V\setminus\{v\}}X_t(\{v,w\})\leq \degU+1,\\
         &\textstyle\sum_{w\in V\setminus\{u\}}X_t(\{u,w\})\leq \degU+1
    \end{aligned} \,\middle|\, 
    X_t(\{v,u\}) = 1,\, U_t=U,\, \mathcal{E} \right] \nonumber \\
    &= \Pr\left[\begin{aligned}
         &\textstyle\sum_{w\in V\setminus\{v,u\}}X_t(\{v,w\})\leq \degU,\\
         &\textstyle\sum_{w\in V\setminus\{u,v\}}X_t(\{u,w\})\leq \degU
    \end{aligned} \,\middle|\, 
    X_t(\{v,u\}) = 1,\, U_t=U,\, \mathcal{E} \right] 
\nonumber\\
    &=\Pr\left[\textstyle\sum_{w\in V\setminus\{v,u\}}X_t(\{v,w\})\leq \degU \,\middle|\, U_t=U,\mathcal{E} \right]\nonumber \\
    &\hspace{1em}\cdotp \Pr\left[ \textstyle\sum_{w\in V\setminus\{u,v\}}X_t(\{u,w\})\leq \degU \,\middle|\, U_t=U,\mathcal{E}
    \right]
\nonumber\\
    &=\Pr\left[\textstyle\sum_{w\in V\setminus\{v,u\}}X_t(\{v,w\})\leq \degU \,\middle|\, U_t=U,\mathcal{E} \right]^2
    \label{tmp20220527b}
\end{align}
 holds where both of the last two equalities follow the fact that $X_t(\{w,w'\})$ for $\{w,w'\}\in E$ are independent. 
Concerning \cref{tmp20220527b},  
  we remark that its expectation satisfies 
\begin{align}
  \E\left[\textstyle\sum_{w\in V\setminus\{v,u\}}X_t(\{v,w\}) \,\middle|\, U_t=U,\mathcal{E}\right] 
  &=\textstyle\sum_{w\in V\setminus\{v,u\}}\E\left[X_t(\{v,w\}) \,\middle|\, U_t=U,\mathcal{E}\right] \nonumber\\
  &=\textstyle\sum_{w\in V\setminus\{v,u\}}\Pr\left[X_t(\{v,w\})=1 \,\middle|\, U_t=U,\mathcal{E}\right] \nonumber\\
  &\leq n\max\{p,1-q\}
\label{tmp20220527c}  
\end{align}
where the last inequality follows the edge-Markovian model \cref{eq:Edge-Markov-prob}. 
Thus, we have 
\begin{align}
    &\Pr\left[\textstyle\sum_{w\in V\setminus\{v,u\}}X_t(\{v,w\})\leq \degU \,\middle|\, U_t=U,\mathcal{E}\right]\nonumber \\
    &\textstyle=
    1-\Pr\left[\sum_{w\in V\setminus\{v,u\}}X_t(\{v,w\})> \degU \,\middle|\, U_t=U,\mathcal{E}\right]\nonumber\\
    &\textstyle\geq 
    1-\Pr\left[\sum_{w\in V\setminus\{v,u\}}X_t(\{v,w\})\geq \degU \,\middle|\, U_t=U,\mathcal{E}\right]\nonumber\\
    &\textstyle=
    1-\Pr\left[\sum_{w\in V\setminus\{v,u\}}X_t(\{v,w\})\geq 2 n\max\{p,1-q\} \,\middle|\, U_t=U,\mathcal{E}\right]
    \hspace{1em}(\mbox{since $\degU = 2n \max\{p,1-q\}$}) \nonumber\\
    &\geq 1-\exp\left(-\frac{n\max\{p,1-q\}}{3}\right)
    \hspace{4.2em}(\mbox{by \cref{lem:Chernoff} (i) using \cref{tmp20220527c}}) \nonumber\\
    &\geq 1-\exp\left(-\frac{\RG}{3}\right)
    \hspace{10em}(\mbox{since $n\max\{p,1-q\}\geq\RG$ by assumption}) 
    \label{eq:RMkey-3}
\end{align}
 hold. 
 By Eqs. \cref{eq:RMkey-2,tmp20220527b,eq:RMkey-3}, 
   we obtain the desired claim \cref{eq:RMkey-4}.

Now, combining \cref{eq:RMkey-1,eq:RMkey-4}, we obtain
\begin{align}
    &\textstyle\Pr\left[M_t(v) \in U_t \,\middle|\, \mathcal{E}\right]\nonumber\\
  &\geq
         \frac{\fair}{\degU+1}
         \sum_{U \in 2^V}\sum_{u\in U}
    \Pr\left[G_t\in \mathcal{G}_{v,u}\,\middle|\, U_t=U, \mathcal{E}\right]
    \Pr\left[U_t=U\,\middle|\, \mathcal{E}\right] 
    &&(\mbox{by \cref{eq:RMkey-1}})\nonumber\\
  &\geq
         \frac{\fair\min\{p,1-q\}\left(1-\exp\left(-\frac{\RG}{3}\right)\right)^2}{\degU+1}
         \sum_{U \in 2^V}\sum_{u\in S}
    \Pr\left[U_t=U \,\middle|\, \mathcal{E}\right] 
    &&(\mbox{by \cref{eq:RMkey-4}})\nonumber\\
  &=
         \frac{\fair\min\{p,1-q\}\left(1-\exp\left(-\frac{\RG}{3}\right)\right)^2}{\degU+1}
        \sum_{U \in 2^V}
         \left(|U| \cdotp \Pr\left[U_t=U\,\middle|\, \mathcal{E}\right] \right)\nonumber\\
  &=
         \frac{\fair\min\{p,1-q\}\left(1-\exp\left(-\frac{\RG}{3}\right)\right)^2}{\degU+1}
   \E\left[|U_t| \,\middle|\, \mathcal{E}\right] 
\label{eq:RMkey-5}
\end{align}
 hold. 
 Finally, we remark that the coefficient of \cref{eq:RMkey-5} satisfies 
\begin{align}
   & \frac{\fair\min\{p,1-q\}\left(1-\exp\left(-\frac{\RG}{3}\right)\right)^2}{\degU+1} \nonumber\\
   &\geq 
    \frac{\fair\min\{p,1-q\}\left(1-\exp\left(-\frac{\RG}{3}\right)\right)^2}{2n \max\{p,1-q\}+1} 
      &&(\mbox{since $\degU=2n \max\{p,1-q\}$ by definition})\nonumber\\
    &\geq \frac{\fair\min\{p,1-q\}\left(1-\exp\left(-\frac{\RG}{3}\right)\right)^2}{2n \max\{p,1-q\}(1+\frac{1}{2\RG})}
      &&(\mbox{since $2n \max\{p,1-q\} \geq 2\RG$ by assumption})\nonumber\\
    &=\frac{\left(1-\exp\left(-\frac{\RG}{3}\right)\right)^2}{2+\frac{1}{\RG}}\frac{\fair}{\frac{\max\{p,1-q\}}{\min\{p,1-q\}}} \frac{1}{n}
\nonumber\\
    &=c^*\frac{\fair}{\rate}\frac{1}{n}
      &&\left(\mbox{where $\dfrac{\max\{p,1-q\}}{\min\{p,1-q\}}= \max\left\{\frac{1-q}{p},\frac{p}{1-q}\right\}=\rate$}\right) \label{eq:RMkey-6}
\end{align}
where $c_*=\dfrac{\left(1-\exp\left(-\frac{\RG}{3}\right)\right)^2}{2+\frac{1}{\RG}}$ is given by~\cref{def:c*}. 
\cref{eq:RMkey-0} is clear from \cref{eq:RMkey-5,eq:RMkey-6}. 
\end{proof}

\section{Proof of Theorem \ref{thm:RM_general}}\label{sec:proof}
 We prove \cref{thm:RM_general} by a version of the token-based analysis developed by \cite{BFKK19}. 
 As a preliminary step, 
   \cref{sec:analytic_framework} introduces definitions for our token-based analysis. 
 \cref{sec:proof-sketch} briefly explains our proof strategy based on the token-based analysis, and 
  the detail of the proof follows.

\subsection{Preliminary for a token-based analysis}
\label{sec:analytic_framework}
\begin{figure*}[t]
    \centering
    \includegraphics[width=16cm]{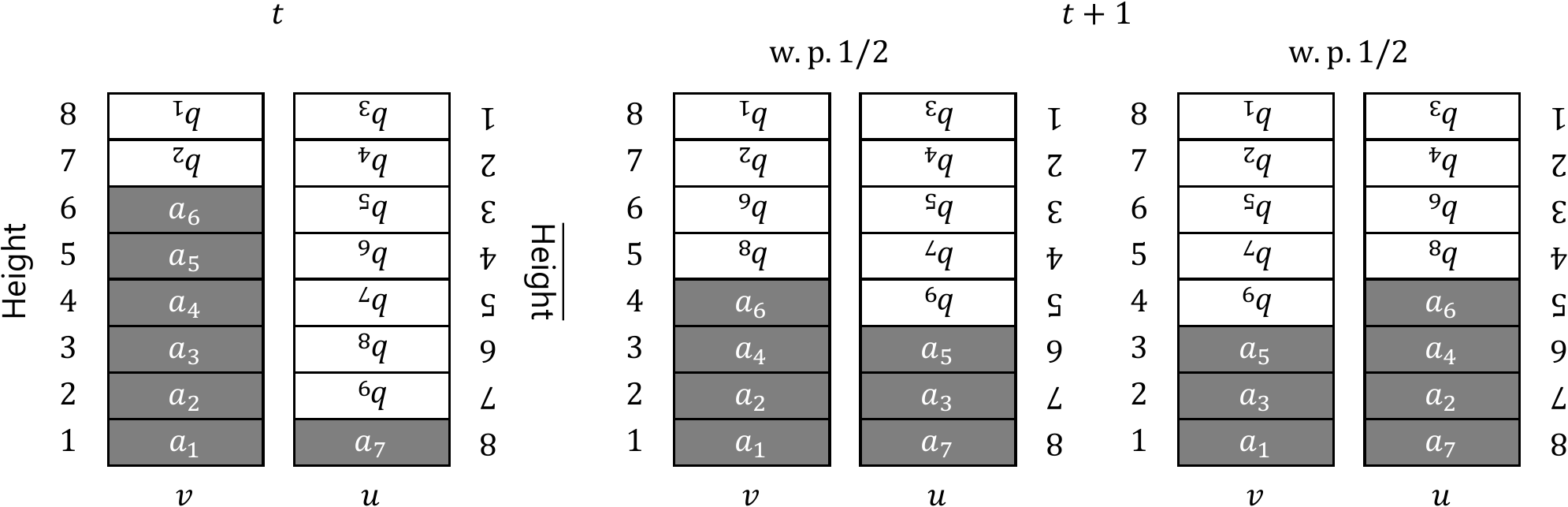}
    \caption{An example of reallocation of tokens and complementary tokens from $\aL_t$ to $\aL_{t+1}$. 
    Here, $\aL_t(v)=6$ and $\aL_t(u)=1$ (left figure), and 
        assume that the total number of tokens $\tL=\sum_{v' \in V} \aL_t(v') =8$, 
           meaning that one more token places another vertex $w$ behind the figure. 
    Suppose $\{u,v\} \in M_t$. 
    Then, either 
      $\aL_{t+1}(v)=4$ and $\aL_{t+1}(u)=3$ (middle figure), or 
      $\aL_{t+1}(v)=3$ and $\aL_{t+1}(u)=4$ (right figure) are obtained with equally probability 1/2. 
 For an example of the individual token's move, 
   token $a_5$ which places $(\Pos_t(a_5), \Hei_t(a_5)) =(v,5)$ at time $t$ (left fig.)
   moves to $(\Pos_{t+1}(a_5), \Hei_{t+1}(a_5)) =(u,3)$ (middle fig.) or $(v,3)$ (right fig.) at time $t+1$. 
 For an example of the complementary token's move, 
   token $b_5$ places $(\bPos_t(b_5), \bHei_t(b_5)) =(u,3)$ at time $t$, 
   where notice that $\Hei_t(b_5) = \tL+1-\bHei(b_5)$ holds. 
 Then, $b_5$   
   moves to $(\bPos_{t+1}(a_5), \bHei_{t+1}(a_5)) =(u,3)$ (middle fig.) or $(v,3)$ (right fig.) at time $t+1$. }
    \label{fig:loads}
\end{figure*}
 The idea of the token-based analysis 
   is to track the {\em place} and {\em height} for each token in an execution of a load-balancing algorithm, 
 where we assume that 
   every token has a unique ID,  
   tokens on a vertex are stacked in a pile, and 
   tokens are orderly reallocated in each time step 
     by the load-balancing algorithm under a refined description of the procedures. 

 Suppose we have an initial token configuration $\aL_0 \in \mathbb{N}^V$ of $\tL=\sum_{v\in V}\aL_0(v)$ tokens. 
 The $\tL$ tokens have distinct labels (i.e., unique IDs) $a_1, a_2, \ldots, a_{\tL}$. 
 For convenience, let $A$ denote the entire set of tokens, i.e., $A = \{a_1,\ldots,a_{\tL}\}$. 
 Every token $a \in A$ is allocated $(\Pos_0(a),\Hei_0(a)) \in V \times \mathbb{N}_{>0}$, 
  where 
  $\Pos_t(a)$ denotes the vertex on which the token $a$ places at time $t$ ($t=0,1,2,\ldots$), and 
  $\Hei_t(a)$ ($1 \leq \Hei_t(a) \leq \aL_t(v)$) represents 
    the ``height'' of the token $a$ in the ``pile'' of $\aL_t(v)$ tokens at vertex $v = \Pos_t(a)$; 
  thus $\{\Hei_t(a') \mid \Pos_t(a')=v\,(a' \in A)\} = \{1,2,\ldots,\aL_t(v)\}$ must hold for any $v \in V$ 
  (see \cref{fig:loads} for example). 

 Then, we define the procedure to update $(\Pos_t(a),\Hei_t(a))$, 
   meaning that it is a refinement of the random-matching algorithm (on an edge-Markovian graph)
    given in \cref{sec:algorithm}.  
 Suppose 
   that $\aL_t$ is the token configuration at time $t$, and 
   that $M_t$ is the random matching given by the random matching algorithm. 
 Let $\{u,v\}$ be an edge of $M_t$. 
 Without loss of generality, we may assume that $\aL_t(v) \geq \aL_t(u)$. 
 If the token $a$ satisfies $\Hei_t(a) \leq \aL_t(u)$, 
  then the token $a$ stays as it is, i.e., $(\Pos_{t+1}(a),\Hei_{t+1}(a))=(\Pos_t(a),\Hei_t(a))$
   (see e.g., token $a_1$ in \cref{fig:loads}). 
 Suppose $\Hei_t(a) > \aL_t(u)$ (see e.g., token $a_5$ in \cref{fig:loads}), 
   which implies $\Pos_t(a) = v$ since $\aL_t(v) \geq \aL_t(u)$. 
 Then, the token $a$ moves to 
\begin{align*}
    \bigl(\Pos_{t+1}(a),\Hei_{t+1}(a)\bigr)=
    \left\{
    \begin{aligned}
    &(v,\aL_t(u)+k) && (\mbox{if $\Hei_t(a)-\aL_t(u) = 2k-1$} &&(k \in \{1,2,\ldots\}))\\
    &(u,\aL_t(u)+k) && (\mbox{if $\Hei_t(a)-\aL_t(u) = 2k$} &&(k \in \{1,2,\ldots\}))
    \end{aligned} \right.
\end{align*}
  in case (i) of \cref{def:simplebalance} (middle in \cref{fig:loads}), while 
\begin{align*}
    \bigl(\Pos_{t+1}(a),\Hei_{t+1}(a)\bigr)=
    \left\{
    \begin{aligned}
    &(v,\aL_t(u)+k) && (\mbox{if $\Hei_t(a)-\aL_t(u) = 2k$} &&(k \in \{1,2,\ldots\}))\\
    &(u,\aL_t(u)+k) && (\mbox{if $\Hei_t(a)-\aL_t(u) = 2k-1$} &&(k \in \{1,2,\ldots\}))
    \end{aligned} \right.
\end{align*}
  in case (ii) (right in \cref{fig:loads})\footnote{
  When $\aL_t(v)-\aL_t(u)$ is even, cases (i) and (ii) of \cref{def:simplebalance} 
    provides the same configuration $\aL_{t+1}(v)$ and $\aL_{t+1}(v)$ for $v$ and $u$. 
  Concerning token reallocation, 
  either way of (i) and (ii) is fine: 
  Each provides the following property (A1) and (A2), 
  and it does not cause any trouble in the following our analysis. }.  
 It is easy to see that 
   this procedure provides the configuration $\aL_{t+1}$ 
   defined by the random matching algorithm.  
 It is also not difficult to see the following facts. 
\begin{observation}\label{obs:a} 
  The function $\Hei_t$ ($t=0,1,2,\ldots$), 
     which is sequentially provided by the above procedure, 
   has the following two properties. 
\begin{enumerate}[label=(A\arabic*)]
    \item \label{prop:monotone} 
      $\Hei_{t}$ is {\em monotone nonincreasing} with respect to $t$, 
        i.e., $\Hei_{t+1}(a) \leq \Hei_t(a)$ for any token $a \in A$ and any time $t = 0,1,2,\ldots$.
    \item \label{prop:average}
     Suppose $\{v,u\} \in M_t$ and $\aL_t(v)\geq \aL_t(u)$. 
     If a token $a \in A$ satisfies both $\Pos_t(a)=v$ and $\Hei_t(a)\geq \aL_t(u)$, 
     then $\Hei_{t+1}(a)-\aL_t(u)=\left\lceil \frac{\Hei_t(a)-\aL_t(u)}{2}\right \rceil$ holds. 
\end{enumerate}
\end{observation}

%
%
%

 Next, 
   we introduce a gadget of complementary tokens, 
      the concept of which is similar to the one used in \cite{BFKK19}, 
   in order to use a symmetric argument to simplify our token-based analysis. 
 Let $\bL_t \in \mathbb{N}^V$ for $t=0,1,2,\ldots$ be defined by $\bL_t(v)\defeq \tL-\aL_t(v)$ for all $v\in V$, 
  where we call $\bL_t$ the {\em configuration of complementary tokens} at time $t$. 
 For convenience, 
  let $\btL\defeq \sum_{v\in V}\bL_0(v)=\tL(n-1)$ denote the total number of complementary tokens, and 
   let $\bave\defeq (1/n)\sum_{v\in V}\bL_0(v)=\tL(n-1)/n$ denote its average.

 It is not difficult to see that 
  if the time series $\aL_0,\aL_1,\aL_2,\ldots$ follows the random matching algorithm then 
  the time series $\bL_0,\bL_1,\bL_2,\ldots$ itself also follows the random matching algorithm 
    with exactly the same matchings $M_t$ ($t=0,1,2,\ldots$). 
 Then, we will define the procedure for complementary tokens, 
   which is essentially the same as the procedure for $A$.  
 We assume that complementary tokens also have distinct labels $b_1, b_2, \ldots, b_{\btL}$, and 
   let $B$ denote the entire set of complementary tokens. 
 Every complementary token $b \in B$ is allocated $(\bPos_0(b),\bHei_0(b)) \in V \times \mathbb{N}_{>0}$
 where 
  $\bPos_t(b)$ denotes the vertex on which the token $b$ places at time $t$ ($t=0,1,2,\ldots$), and 
  $\bHei_t(b)$ ($1 \leq \bHei_t(a) \leq \bL_t(v)$) represents 
    the ``height'' of the token $b$ at $v = \bPos_t(b)$. 
 For convenience, 
   we define $\Hei_t(b) \defeq \tL+1-\bHei_t(b)$ for $b \in B$. 
 Then, $\aL(v)+1 \leq \Hei_t(b) \leq \tL$ holds. 
 It looks that complementary tokens are stacked on tokens $\aL(v)$, in the inverse order of $\Hei_t(b)$. 
 Thus, we call  $\bHei_t$ {\em inverted height}. 
  
 Then, we define the procedure of the random-matching algorithm for complementary tokens. 
 Let $\aL_t$ and $\bL_t$ be respectively the configurations of tokens and complementary tokens at time $t$. 
 Suppose  that $M_t$ is a random matching, and $\{u,v\} \in M_t$. 
 Without loss of generality we may assume 
   $\aL_t(v) \geq \aL_t(u)$. 
 Then  $\bL_t(v) \leq \bL_t(u)$. 
 If the token $b$ satisfies $\bHei_t(a) \leq \bL_t(u)$, 
  then the token $b$ stays as it is i.e., $(\bPos_{t+1}(b),\bHei_{t+1}(b))=(\bPos_t(a),\bHei_t(b))$. 
 Suppose $\bHei_t(b) > \bL_t(u)$. 
 Notice that $\bPos(b) = u$. 
 Then, the token $b$ moves to 
\begin{align*}
    \bigl(\bPos_{t+1}(b),\bHei_{t+1}(b)\bigr)=
    \left\{
    \begin{aligned}
    &(v,\bL_t(v)+k) && (\mbox{if $\bHei_t(b)-\bL_t(v)=2k$}  &&(k \in \{1,2,\ldots\})) \\
    &(u,\bL_t(v)+k) && (\mbox{if $\bHei_t(b) -\bL_t(v)=2k-1$} &&(k \in \{1,2,\ldots\}))
    \end{aligned}\right.
\end{align*}
in case (i) of \cref{def:simplebalance} (see also middle in \cref{fig:loads}), while 
\begin{align*}
    \bigl(\bPos_{t+1}(b),\bHei_{t+1}(b)\bigr)=
    \left\{
    \begin{aligned}
    &(v,\bL_t(v)+k) && (\mbox{if $\bHei_t(b)-\bL_t(v)=2k-1$}  &&(k \in \{1,2,\ldots\})) \\
    &(u,\bL_t(v)+k) && (\mbox{if $\bHei_t(b) -\bL_t(v)=2k$} &&(k \in \{1,2,\ldots\}))
    \end{aligned}\right.
\end{align*}
  in case (ii) of \cref{def:simplebalance} (see also right in \cref{fig:loads}). 
 We can see that 
   this procedure provides the configuration $\bL_{t+1}$ 
   defined by the random matching algorithm.  
 The following observation is essentially the same as Observation \ref{obs:a}.  
\begin{observation} \label{obs:b} 
  The function $\bHei_t$ ($t=0,1,2,\ldots$), 
     which is sequentially provided by the above procedure, 
   has the following two properties. 
\begin{enumerate}[label=(B\arabic*)]
    \item \label{prop:monotoneb} 
      $\bHei_{t}$ is {\em monotone nonincreasing} with respect to $t$, 
        i.e., $\bHei_{t+1}(b) \leq \bHei_t(b)$ for any token $b \in B$ and any time $t = 0,1,2,\ldots$.
    \item \label{prop:averageb} 
     Suppose $\{v,u\} \in M_t$ and $\bL_t(v)\leq \bL_t(u)$. 
     If a token $b \in B$ satisfies both $\bPos_t(b)=v$ and $\bHei_t(b)\geq \bL_t(u)$, 
     then $\bHei_{t+1}(b)-\bL_t(u)=\left\lceil \frac{\bHei_t(b)-\bL_t(u)}{2}\right \rceil$ holds. 
\end{enumerate}
\end{observation}

 From Observations~\ref{obs:a} and \ref{obs:b}, 
  we can show the following lemma, which forms the basis of our analysis in \cref{sec:proof}.  
 \cref{lem:halved} will be used as $x=\ave$, $\bave$, $\rave$ and $\brave$ in the following sections. 
\begin{lemma}
\label{lem:halved}
 Let $x \in \mathbb{R}$ be an arbitrary real. The following holds for any $t \in \mathbb{N}$. 
\begin{enumerate}[label=(\roman*)]
    \item \label{lab:halved} 
      Let $a \in A$ be a token satisfying $\Hei_t(a) \geq x$. 
      Suppose that the vertex $\Pos_t(a)$ is matched with $v \in V$ in $M_t$, i.e., $v= M_t(\Pos_t(a))$. 
      If $\aL_t(v) \leq x$ then $\Hei_{t+1}(a)-x\leq \left\lceil \frac{\Hei_t(a)-x}{2}\right \rceil$. 
    \item \label{lab:halved_i} 
      Let $b \in B$ be a complementary token satisfying $\bHei_t(b_j)\geq x$. 
      Suppose that the vertex $\bPos_t(b)$ is matched with $u \in V$ in $M_t$, i.e., $u= M_t(\bPos_t(b))$. 
      If $\bL_t(u)\leq x$ then, $\bHei_{t+1}(b)-x \leq \left\lceil \frac{\bHei_t(b)-x}{2}\right \rceil$.  
\end{enumerate}
\end{lemma}

\begin{proof}
We prove \ref{lab:halved}. 
To begin with, we see that  
\begin{align}
    \Hei_{t+1}(a)-x
    &=\Hei_{t+1}(a)-\aL_t(v)+\aL_t(v)-x \nonumber\\
    &=\left\lceil \frac{\Hei_{t}(a)-\aL_t(v)}{2}\right \rceil+\aL_t(v)-x
    \label{eq:halved1}
\end{align}
 holds by Observations~\ref{obs:a}. 
 If $\Hei_{t}(a)-\aL_t(v)$ is even then 
\begin{align*}
    \eqref{eq:halved1}
    &=\frac{\Hei_{t}(a)-\aL_t(v)}{2}+\aL_t(v)-x \\
    &=\frac{\Hei_{t}(a) + \aL_t(v)-2x}{2} \\
    &= \frac{\Hei_{t}(a)-x}{2}+\frac{\aL_t(v)-x}{2}\\ 
    &\leq \frac{\Hei_{t}(a)-x}{2} &&(\mbox{since $\aL_t(v) \leq x$ by hypothesis}) \\
    &\leq \left\lceil \frac{\Hei_{t}(a)-x}{2} \right\rceil 
\end{align*}
  holds, and we obtain the claim in this case. 

 Suppose $\Hei_{t}(a)-\aL_t(v)$ is odd. 
 Then, we have
\begin{align}
    \eqref{eq:halved1}
    &= \frac{\Hei_{t}(a)-\aL_t(v)}{2}+\frac{1}{2}+\aL_t(v)-x \nonumber\\
    &=\frac{\Hei_{t}(a) + \aL_t(v)-2x+1}{2} \nonumber\\
    &= \frac{\Hei_{t}(a)-x}{2}+\frac{\aL_t(v)-x+1}{2} 
    \label{eq:halved2}
\end{align}
 holds. 
 If $\aL_t(v)-x \leq -1$ then 
   $\eqref{eq:halved2} \leq \frac{\Hei_{t}(a)-x}{2} \leq \left\lceil \frac{\Hei_{t}(a)-x}{2} \right\rceil$ 
   holds, and we obtain the claim in this case. 

 The remaining case, that is 
   $\Hei_{t}(a)-\aL_t(v)$ is odd and  $-1 < \aL_t(v)-x \leq 0$. 
 Since $-1 < \aL_t(v)-x \leq 0$, 
\begin{align}
   \left\lceil \frac{\Hei_{t}(a)-\aL_t(v)-1}{2}\right  \rceil
    < \left\lceil \frac{\Hei_{t}(a)-x}{2}\right \rceil
   \leq \left\lceil \frac{\Hei_{t}(a)-\aL_t(v)}{2}\right  \rceil
    \label{eq:halved3}
\end{align}
 hold. 
 Since $\Hei_{t}(a)-\aL_t(v)$ is odd, 
  the strict inequality and the integrality of \eqref{eq:halved3} imply 
\begin{align}
   \left\lceil \frac{\Hei_{t}(a)-\aL_t(v)}{2}\right  \rceil
   = \left\lceil \frac{\Hei_{t}(a)-x}{2}\right \rceil
    \label{eq:halved4}
\end{align}
holds. Then, 
\begin{align*}
    \eqref{eq:halved1}
    &= \left\lceil \frac{\Hei_{t}(a)-x}{2}\right \rceil +\aL_t(v)-x
      &&(\mbox{by \eqref{eq:halved4}})  \\
    &\leq \left\lceil \frac{\Hei_{t}(a)-x}{2} \right\rceil 
      &&(\mbox{since $\aL_t(v) \leq x$ by hypothesis}) 
\end{align*}
 holds. We obtain (i). 
 The proof of (ii) is similar. 
\end{proof}

\subsection{Proof strategy}\label{sec:proof-sketch}
 Roughly speaking, 
   \cref{thm:RM_general} claims 
  $\ave-2 < \aL_{\Tbal}(v) < \ave+2$ for any $v \in V$ w.h.p. 
  (see \cref{lem:logic} in \cref{sec:final} in precise).  
 We prove the claim in two phases: 
 In Phase I (see \cref{sec:phasei}), we prove 
   at least one of 
     $ \ave-1 \leq \min_{v \in V} \aL_{T_1}(v) $ or 
     $ \max_{v \in V}\aL_{T_1}(v) \leq  \ave+1$ 
    holds w.h.p. for a sufficiently large $T_1$ (see \cref{lem:PhaseI}, for more detail). 
 In Phase II (see \cref{sec:phaseii}), we prove that 
  if $ \ave-1 \leq \min_{v \in V}\aL_{T_1}(v) $ 
    then $ \max_{v \in V} \aL_{T_1+T_2}(v) < \ave+2$ also holds w.h.p., as well as 
  if $ \max_{v \in V}\aL_{T_1}(v) \leq \ave+1 $ 
    then $ \ave-2 < \min_{v \in V} \aL_{T_1+T_2}(v)$ 
    also holds w.h.p. for a sufficiently large $T_2$ (see \cref{lem:PhaseII}, for more detail). 
 We remark that
   $\min_{v \in V} \aL_t(v)$ and $\max_{v \in V} \aL_t(v)$ are respectively 
    monotone non-decreasing/non-increasing with respect to $t$, 
  which imply $ \ave-2 < \aL_{T_1+T_2}(v) < \ave+2$ holds for any $v \in V$.

We prove the claims by the token-based analysis introduced in \cref{sec:analytic_framework}, 
using \cref{lem:halved,lem:prob_RM}. 
In the following arguments, 
  we assume that 
   an edge-Markovian evolving graph $\bmG(n,p,q)$ satisfies 
   $\max\{p,1-q\} \geq \RG/n$ for a constant $\RG > 0$, and 
  that a random matching algorithm satisfies the $\fair$ condition \cref{def:matching}, 
 according to the hypothesis of \cref{thm:RM_general}. 
    For a technical reason, 
   we also assume $\Delta \geq 2$ (recall \cref{def:Delta}); 
   otherwise it is already balanced, i.e., $\Tbal=0$, and the theorem is trivial.

\subsection{Phase I}
\label{sec:phasei}
Let us begin with Phase I analysis. This section establishes the following lemma. 
\begin{lemma}[Phase I]
\label{lem:PhaseI}
 Let $T_1 \geq \frac{36\rate}{c^*\fair}\log(\frac{\Delta n}{\varepsilon})$
   for $\varepsilon$ ($0<\varepsilon<1/4$).
Then,  
\begin{align}
 \Pr\left[[\max_{v\in V}\aL_{T_1}(v)\leq \ave+1] \vee [\max_{v\in V}\bL_{T_1}(v)\leq \bave+1]\right] 
   \geq 1-\varepsilon^2
\label{lem:PhaseIeq}
\end{align}
    holds.\footnote{
  We here just remark that  
 $[\max_{v\in V}\bL(v)\leq \bave+1]$ is equivalent to $[\min_{v \in V}L(v) \geq  \ave-1]$. 
 See Observation~\ref{obs:PhaseI} for detail. 
   }
\end{lemma}
\begin{proof}
 We prove the claim by the token-based analysis given in \cref{sec:analytic_framework}. 
 For convenience, 
  let $A^+=\{a \in A \mid \Hei_0(a) > \ave+1\}$ and 
      $B^+=\{b \in B \mid \bHei_0(b) > \bave+1\}$. 
 Firstly, we remark that 
   $\max_{a \in A^+}\Hei_t(a) \leq \ave+1$ implies (if and only if, in fact)
   $\max_{v\in V}\aL_t(v) \leq \ave+1$ 
  since $\Hei_t(a)$ is monotone non-increasing with respect to $t$,  
    by Observation~\ref{obs:a} (A1). 
 Remark $\max_{b \in B^+}\bHei_t(b) \leq \bave+1$ implies $\max_{v\in V}\bL_t(v) \leq \bave+1$ as well.  
 Thus, our desired event is rephrased by 
\begin{align}
&\left[[\max_{v\in V}\aL_{T_1}(v)\leq \ave+1] \vee [\max_{v\in V}\bL_{T_1}(v)\leq \bave+1]\right] \nonumber\\
&\Leftrightarrow
  \left[\max_{a \in A^+}[\Hei_{T_1}(a)\leq \ave+1] \vee 
  \max_{b \in B^+}[\bHei_{T_1}(b)\leq\bave+1]\right] \nonumber\\
&\Leftrightarrow
  \left[\bigwedge_{a \in A^+}[\Hei_{T_1}(a)- \ave\leq1] \vee 
  \bigwedge_{b \in B^+}[\bHei_{T_1}(b) - \bave \leq1]\right] \nonumber\\
&\Leftarrow
  \bigwedge_{(a,b) \in A^+ \times B^+}
  \left[[\Hei_{T_1}(a)-\ave\leq 1] \vee [\bHei_{T_1}(b)-\bave\leq 1]\right]
\label{tmp20220523b}
\end{align}
  where the last converse implication is for some technical reason of the arguments below\footnote{
  We are going to use a dichotomy in \cref{eq:W_lower}, below. 
  }. 

 Roughly speaking, \cref{lem:halved} claim that 
  $\Hei_{t+1}(a)-\ave$ is reduced to almost a half of $\Hei_t(a)-\ave$ 
  when the vertex $\Pos_t(a)$ is matched with a vertex $v$ satisfying $\aL_t(v)\leq \ave$ in $M_t$, 
   in case that  $\Hei_t(a)-\ave >0$.  
 For convenience,  let 
  $S(\aL_t) \defeq \{v\in V \mid \aL_t(v)\leq \ave\}$ and 
  $\overline{S}(\aL_t) \defeq \{v\in V \mid \aL_t(v)\geq \ave\}=\{v\in V \mid \bL_t(v)\leq \bave\}$ 
  for each time $t=0,1,2,\ldots$. 
 Then, let 
\begin{align}
Y_t(a,b) = \begin{cases} 
  1 & (\mbox{if } [M_t(\Pos_t(a))\in S(\aL_t)] \vee [M_t(\Pos_t(b)) \in \overline{S}(\aL_t)]) \\
  0 & (\mbox{otherwise}) 
\end{cases}
\end{align}
 for $(a,b) \in A^+ \times B^+$, 
  where we remark that $[M_t(\Pos_t(a))\in S(\aL_t)]$ means 
    ``$\exists u \in  S(\aL_t)$ such that $\{u,v\} \in M_t$ where $v=\Pos_t(a)$.''
 In fewer words,  $Y_t(a,b)$ denotes the indicator random variable 
  of a desired matching at time $t$. 
 To circumvent the effect of ceiling function in \cref{lem:halved}, 
  we remark a fact  that 
    $\left\lceil\frac{\lceil \frac{x}{2} \rceil}{2}\right\rceil \leq \frac{x}{2}$ holds\footnote{
  Proof: 
  To begin with,  
           $\left\lceil\frac{\lceil \frac{x}{2} \rceil}{2}\right\rceil 
         < \left\lceil\frac{\frac{x}{2}+1}{2}\right\rceil
         = \lceil\frac{x+2}{4}\rceil$
         holds. 
 Since both of the left-hand-side and the right-hand-side are integers, 
  the strict inequality implies 
           $\left\lceil\frac{\lceil \frac{x}{2} \rceil}{2}\right\rceil \leq \lceil\frac{x+2}{4}\rceil-1$. 
 Next, 
         $\lceil\frac{x+2}{4}\rceil
         < \frac{x+2}{4}+1
         =\frac{x+6}{4}
         =\frac{x}{2} + \frac{-x+6}{4} 
         \leq \frac{x}{2} + 1$
         where the last inequality follows the assumption $x \geq 2$.  
 Now it is not difficult to see that 
           $\left\lceil\frac{\lceil \frac{x}{2} \rceil}{2}\right\rceil \leq \frac{x}{2}$. 
    } for $x \geq 2$, 
  which implies 
   $\Hei(a)-\ave$ is reduced to half or less if the event $[M_t(\Pos_t(a))\in S(\aL_t)]$ occurs TWICE. 
 Clearly,  
    $\lceil \frac{x}{2} \rceil \leq 1$ holds for $1 \leq x  \leq2$, 
    thus we get $\Hei(a)_t-\ave \leq 1$ 
      if we got the event $[M_t(\Pos_t(a))\in S(\aL_t)]$ 
      at most  $2\log_2 \Delta + 1 \leq 3 \log_2 \Delta$ times\footnote{
    Let $y_{2i} = \frac{x}{2^i}$ for $i=1,2,\ldots$, then $y_{2(\lceil \log_2 x \rceil -1)} = \frac{2x}{2^{\lceil \log_2 x \rceil}} \leq 2$. 
    Note $2(\lceil \log_2 x \rceil -1) + 1 = 2\lceil \log_2 x \rceil -1 \leq  2 \log_2 x + 1$. 
   }${^,}$\footnote{
       Here we use the technical assumption that $\Delta \geq 2$ 
       for the inequality $2\log_2 \Delta + 1 \leq 3 \log_2 \Delta$. 
       },  
   where we remark $\Hei_0(a)-\ave \leq  \max_{v\in V}\aL_0(v)-\min_{v\in V}\aL_0(v) =  \Delta$. 
 $[\bHei_{T_1}(b)-\bave\leq 1] \leq 1$ is as well. 
 These imply that 
   it is sufficient to get $[\Hei_{T_1}(a)-\ave\leq 1] \vee [\bHei_{T_1}(b)-\bave\leq 1]$ that 
    $\sum_{t=0}^{T_1-1}Y_t(a,b) \geq 6 \log_2 \Delta$ holds.
In summary, we obtain 
\begin{align}
&\Pr\left[[\max_{v\in V}\aL_T(v)\leq \ave+1] \vee [\max_{v\in V}\bL_T(v)\leq \bave+1]\right] \nonumber\\
&\geq
\Pr\left[
\bigwedge_{(a,b) \in A^+ \times B^+}
\left[[\Hei_{T_1}(a)-\ave\leq 1] \vee [\bHei_{T_1}(b)-\bave\leq 1]\right] 
 \right] 
 &&(\mbox{by \cref{tmp20220523b}})
 \nonumber\\
&\geq
\Pr\left[\bigwedge_{(a,b) \in A^+ \times B^+}
\left[\sum_{t=0}^{T_1-1}Y_t(a,b)\geq 6 \log_2 \Delta \right]\right]  
&&\left(\begin{array}{l}\mbox{by \cref{lem:halved} with the}\\\mbox{above argument}\end{array}\right)
\nonumber\\
&=
1-\Pr\left[\neg\bigwedge_{(a,b) \in A^+ \times B^+}
\left[\sum_{t=0}^{T_1-1}Y_t(a,b)\geq 6 \log_2 \Delta \right]\right]\nonumber\\
&=
1-\Pr\left[\bigvee_{(a,b) \in A^+ \times B^+}
\left[\sum_{t=0}^{T_1-1}Y_t(a,b) < 6 \log_2 \Delta \right]\right]\nonumber\\
&\geq
1-\sum_{(a,b) \in A^+ \times B^+}
\Pr\left[\sum_{t=0}^{T_1-1}Y_t(a,b) <6 \log_2 \Delta \right]
&&(\mbox{union bound}) 
    \label{eq:W_lower2}
\end{align}
and the remaining task is to prove \eqref{eq:W_lower2} is enough large. 

 To evaluate $\Pr[\sum_{t=0}^{T_1-1}Y_t(a,b) <6 \log_2 \Delta]$,  
   we estimate the probability of $Y_t(a,b)=1$ for $t=0,\ldots,T_1-1$. 
 Notice that the event [$Y_t(a,b)=1$] definitely depends 
  on the history of the execution $\mathcal{E}_t$.  
 We give a lower bound of the probability of $Y_t(a,b)=1$ independent of $\mathcal{E}_t$. 
 To be precise, 
\begin{align}
    &\Pr\left[Y_t(a,b)=1 \,\middle|\, \mathcal{E}_t \right] \nonumber \\
    &= \Pr\left[ 
       [M_t(\Pos_t(a)) \in S(\aL_t) ]\vee [M_t(\Pos_t(b)) \in \overline{S}(\aL_t)] 
       \,\middle|\, \mathcal{E}_t
      \right]\nonumber \\
    &\geq \max\left\{
         \Pr\left[M_t(\Pos_t(a)) \in S(\aL_t)  \,\middle|\, \mathcal{E}_t\right],
         \Pr\left[M_t(\Pos_t(b)) \in \overline{S}(\aL_t) \,\middle|\, \mathcal{E}_t\right]
    \right\}\nonumber \\
    &\geq \frac{c_*\fair}{\rate} 
      \max\left\{\E\left[\frac{|S(\aL_t)|}{n} \,\middle|\, \mathcal{E}_t\right],\, 
                     \E\left[\frac{|\overline{S}(\aL_t)|}{n} \,\middle|\, \mathcal{E}_t\right]\right\} 
     \hspace{4em}(\mbox{by \cref{lem:prob_RM}}) 
     \label{ineq:prob_RM}\\
    &\geq  \frac{c_*\fair}{2\rate} 
     \hspace{4em}(\mbox{because $|S(\aL_t)|+|\overline{S}(\aL_t)| \geq n$, by the definitions}) 
    \label{eq:W_lower}
\end{align}
 hold\footnote{
  We emphasize that the inequality \cref{ineq:prob_RM} (and similarly \cref{eq:Z_lower2}, appearing later) 
    is the heart of the analysis of the paper: 
 Due to the edge-Markovian model, 
  $M_t(\Pos_t(a))$ and $S(\aL_t)$ (or $M_t(\Pos_t(b))$ and $\overline{S}(\aL_t)$, as well) 
    may have a correlation through the history $\mathcal{E}_t$, while  
  \cref{lem:prob_RM} proves a history-independent lower-bound. 
 }.

 Then, we evaluate $\Pr[\sum_{t=0}^{T_1-1}Y_t(a,b) <6 \log_2 \Delta]$. 
 For convenience, 
 let $Z_t$ ($t=0,1,\ldots,T_1-1$) be independent binary random variables 
   such that $\Pr[Z_t=1]= \frac{c_*\fair}{2\rate}$,  and 
 let $Z=\sum_{t=0}^{T_1-1}Z_t$. 
 Then, it is not difficult to see that 
\begin{align}
  \Pr\left[\sum_{t=0}^{T_1-1}Y_t(a,b) < 6\log_2 \Delta  \right] 
  \leq \Pr\left[\sum_{t=0}^{T_1-1}Y_t(a,b)\leq 6\log_2 \Delta \right] 
  \leq \Pr\left[Z\leq 6 \log_2 \Delta \right] 
\label{tmp20220604a}
\end{align}
 holds.
Note that 
\begin{align}
 \E[Z] 
  &= T_1 \Pr[Z_t=1] \nonumber \\ 
  &\geq\frac{36\rate}{c_*\fair}\log(\Delta n / \varepsilon) \cdotp \frac{c_*\fair}{2\rate} 
    &&(\mbox{since $T_1 \geq \frac{36\rate}{c_*\fair}\log(\Delta n / \varepsilon)$})\nonumber\\ 
  & \geq 18\log(\Delta n/\varepsilon) \label{tmp20220524a} \\
  & \geq 12\log_2 \Delta   &&(\mbox{since $n \geq 1$, $\varepsilon<1$ and $\log_2{\rm e} \leq 1.5$})  \label{tmp20220524b}
\end{align}
holds. 
Then, we have
\begin{align}
  \Pr\left[\sum_{t=0}^{T_1-1}Y_t(a,b)\leq 6\log_2 \Delta \right] 
  &\leq \Pr\left[Z\leq 6\log_2 \Delta \right] &&(\mbox{by \eqref{tmp20220604a}})\nonumber\\
  &\leq \Pr\left[Z\leq \frac{\E[Z]}{2}\right] &&(\mbox{by \eqref{tmp20220524b}})\nonumber\\
  &\leq \exp\left(-\frac{(\frac{1}{2})^2\E[Z]}{2}\right) &&(\mbox{by \cref{lem:Chernoff} (ii)})\nonumber\\
  &= \exp\left(-\frac{\E[Z]}{8}\right) \nonumber\\
  &\leq \exp\left(-2\log (\Delta n / \varepsilon)\right)&&(\mbox{by \eqref{tmp20220524a}})\nonumber\\
  &= \left(\frac{\varepsilon}{\Delta n}\right)^2. 
    \label{eq:number_of_events}
\end{align}

Finally, 
\begin{align*}
\eqref{eq:W_lower2}
&=
1-\sum_{(a,b) \in A^+ \times B^+}
\Pr\left[\sum_{t=0}^{T_1-1}Y_t(a,b) <6 \log_2 \Delta \right] \\
&\geq
1-\sum_{(a,b) \in A \times B}\left(\frac{\varepsilon}{\Delta n}\right)^2 &&(\mbox{by \eqref{eq:number_of_events}})\\
&\geq
1-(\Delta n)^2\left(\frac{\varepsilon}{\Delta n}\right)^2 \\
&=
1-\varepsilon^2
\end{align*}
and we obtain the claim. 
\end{proof}

Before going to Phase II, we give the following remark concerning \cref{lem:PhaseI}. 
\begin{observation}
\label{obs:PhaseI}
 About the left-hand-side of \eqref{lem:PhaseIeq} in \cref{lem:PhaseI}, 
\begin{align}
 &\Pr\left[[\max_{v\in V}\aL_{T_1}(v)\leq \ave+1] \vee [\max_{v\in V}\bL_{T_1}(v)\leq \bave+1]\right] \nonumber\\
 &=\Pr\left[[\max_{v\in V}\aL_{T_1}(v)\leq \ave+1] \vee [\min_{v\in V}\aL_{T_1}(v)\geq \ave-1]\right] \label{obs:PhaseIb}\\
 &=\Pr\left[[\min_{v\in V}\bL_{T_1}(v)\geq \bave-1] \vee [\min_{v\in V}\aL_{T_1}(v)\geq \ave-1]\right] \label{obs:PhaseIc}
\end{align}
hold. 
\end{observation}
\begin{proof}
 By definition of $\bL_t(v) = \tL - \aL_t(v)$ and a fact $\bave=\tL-\ave$, 
   we remark 
\begin{align}
 &[\max_{v\in V}\aL_{T_1}(v)\leq \ave+1] \Leftrightarrow [\min_{v\in V}\bL_{T_1}(v)\geq \bave-1],\hspace{1em} \mbox{as well  as} 
 \label{tmp:20220609a}\\
 &[\max_{v\in V}\bL_{T_1}(v)\leq \bave+1] \Leftrightarrow [\max_{v\in V}\aL_{T_1}(v)\geq \ave-1]
 \label{tmp:20220609b}
\end{align}
 hold. 
 In fact, concerning \cref{tmp:20220609b}, 
\begin{align*}
 \max_{v\in V} \bL_{T_1}(v)
 &= \max_{v\in V}(\tL-\aL_{T_1}(v))
 = \tL+ \max_{v\in V}(-\aL_{T_1}(v))
 = \tL- \min_{v\in V}\aL_{T_1}(v) \\
 \bave+1
 &= \tL-\ave +1
\end{align*}
 hold, which implies 
\begin{align*}
 [\max_{v\in V}\bL_{T_1}(v)\leq \bave+1] 
 &\Leftrightarrow 
  [\tL- \min_{v\in V}\aL_{T_1}(v) \leq \tL-\ave+1] \\
 &\Leftrightarrow 
  [\min_{v\in V}\aL_{T_1}(v)\geq \ave-1]
\end{align*}
 and we obtain  \cref{tmp:20220609b} which implies \cref{obs:PhaseIb}. 
 \cref{tmp:20220609b} is similar, and hence \cref{obs:PhaseIc} holds, too.
\end{proof}
 As we stated in the proof strategy in~\cref{sec:proof-sketch}, 
   we got \cref{obs:PhaseIb} is enough large by \cref{lem:PhaseI}. 
 In the following sections, we will use  \cref{lem:PhaseI} in the form of \cref{obs:PhaseIc}.

\subsection{Phase II}\label{sec:phaseii}
 By \cref{lem:PhaseI} and Observation~\ref{obs:PhaseI},
  we got a situation that 
  $[\min_{v\in V}\aL(v)\leq \ave+1] \vee [\min_{v\in V}\bL(v)\leq \bave+1]$ w.h.p. in Phase I. 
 We prove the following lemma as Phase II analysis, 
   in mind the Markov property of the execution of the random matching algorithm on an edge-Markovian graph. 
\begin{lemma}[Phase II]
\label{lem:PhaseII}
  Let $T_2 \geq \frac{54\rate}{c^*\fair} \log(\frac{\Delta n}{\varepsilon})$ 
    for $\varepsilon$ ($0<\varepsilon<1/4$).
\begin{enumerate}[label=(\roman*)]
    \item \label{lab:phaseIIa}
    If $\min_{v\in V}\aL_0(v)\geq \ave-1$ then 
      $\Pr[\max_{v\in V}\aL_{T_2}(v)\leq \rave+1] \geq 1-\varepsilon^2$. 
    \item \label{lab:phaseIIb}
    If $\min_{v\in V}\bL_0(v)\geq \bave-1$ then 
      $\Pr[\max_{v\in V}\bL_{T_2}(v)\leq \brave+1] \geq 1-\varepsilon^2$. 
\end{enumerate}
\end{lemma}
Before the proof, 
 we remark that 
  \cref{lem:PhaseII} implies that 
\begin{align}
  \Pr\left[\begin{aligned}
         &[\ave-1\leq \aL_{T_2}(v)\leq \rave+1]\\
         &\vee [\bave-1\leq \bL_{T_2}(v) \leq \brave+1]
    \end{aligned}
    \,\middle|\,
  \begin{aligned}
         &[\min_{u\in V}\aL_0(u)\geq \ave-1]\\
         &\vee[\min_{u\in V}\bL_0(u)\geq \bave-1]
    \end{aligned} 
    \right]
\geq 1-\epsilon^2
\label{lem:PhaseIIb}
\end{align}
 holds for any $v \in V$,  
  where we remark 
   $[\min_{u\in V}\aL_0(u)\geq \ave-1]$ implies $[\min_{u\in V}\aL_t(u)\geq \ave-1]$ for $t \geq 0$ 
    by the monotone non-decreasing property of $\aL_t(v)$ with respect to $t$, $\bL_t(v)$ as well. 
 The \cref{lem:PhaseIIb} is the goal of \cref{lem:PhaseII} 
   as we stated in the proof strategy in~\cref{sec:proof-sketch}. 
 In \cref{sec:final}
   we will finalize the proof of \cref{thm:RM_general} based on  \cref{lem:PhaseI,lem:PhaseII}. 

\begin{proof}[Proof of \cref{lem:PhaseII} (Phase II)]
 We prove \ref{lab:phaseIIa} by the token-based analysis give in \cref{sec:analytic_framework}. 
 For convenience, 
  let $A'=\{a \in A \mid \Hei_{0}(a) > \rave\}$. 
 Then, it is easy to see that 
\begin{align}
 \left[\max_{v\in V}\aL_{T_2}(v)\leq \rave+1\right] 
 &\Leftrightarrow
   \left[\max_{a \in A'} \Hei_{T_2}(a) \leq \rave + 1 \right] \nonumber \\
 &\Leftrightarrow
   \left[\bigwedge_{a \in A'}\left[\Hei_{T_2}(a) - \rave \leq 1\right] \right] 
\label{tmp20220608a}
\end{align}
 holds. 

In a similar way as \cref{lem:PhaseI}, 
 let $S'(\aL_t)\defeq \{v\in V \mid \aL_t(v)\leq \rave\}$, and 
 let 
\begin{align}
Y'_t(a) = \begin{cases} 
  1 & (\mbox{if } M_t(\Pos_t(a))\in S'(\aL_t))\\
  0 & (\mbox{otherwise}) 
\end{cases}
\label{tmp20220608c}
\end{align}
 for $a \in A'$, 
 meaning that $Y'_t(a)$ is the indicator random variable of a desired matching at time $t$. 
 Then, we see that 
  it is sufficient for $\Hei_{T_2}(a) - \rave \leq 1$ 
  that $\sum_{t=0}^{T_2-1} Y'_t(a)\geq 3 \log_2 \Delta$ holds, 
    by \cref{lem:halved} 
    in a similar way as \cref{lem:PhaseI}. 
 Precisely,  
\begin{align}
 \Pr\left[\max_{v\in V}\aL_{T_2}(v)\leq \rave+1\right] 
 &=
   \Pr\left[\bigwedge_{a \in A'}\left[\Hei_{T_2}(a) - \rave \leq 1\right] \right] 
 &&(\mbox{by \eqref{tmp20220608a}})\nonumber\\
 &= \Pr\left[\bigwedge_{a \in A'} \sum_{t=0}^{T_2-1} Y'_t(a)\geq 3 \log_2 \Delta \right] 
 &&\left(\begin{array}{l}\mbox{by \cref{lem:halved} with the}\\\mbox{above argument}\end{array}\right)\nonumber\\
 &= 1- \Pr\left[\bigvee_{a \in A'} \sum_{t=0}^{T_2-1} Y'_t(a) < 3 \log_2 \Delta \right]  \nonumber\\
 &= 1- \sum_{a \in A'} \Pr\left[ \sum_{t=0}^{T_2-1} Y'_t(a) < 3 \log_2 \Delta \right]  
  &&(\mbox{union bound})
\label{tmp20220608b}
\end{align}
 holds. 

 To evaluate $\Pr[\sum_{t=0}^{T_2-1}Y'_t(a) <3 \log_2 \Delta]$,  
  we estimate the probability of $Y'_t(a)=1$ 
   by giving a lower bound independent of the history of execution $\mathcal{E}_t$. 
 Then,  
\begin{align}
   &\Pr\left[Y'_t(a_i)=1 \mid \mathcal{E}_t \right]\nonumber \\
    &=\Pr\left[M_t(\Pos_t(a))\in S'(\aL_t) \,\middle|\, \mathcal{E}_t \right] 
        && (\mbox{by \cref{tmp20220608c}}) \nonumber\\
    &\geq \frac{c_*\fair}{\rate} \E\left[\frac{|S'(\aL_t)|}{n} \,\middle|\, \mathcal{E}_t \right] 
         && (\mbox{by \cref{lem:prob_RM}}) 
    \label{eq:Z_lower2}\\
    &\geq \frac{c_*\fair}{3\rate} 
         && (\mbox{since $|S'(\aL_t)| \geq n/3$ by \cref{lem:propset}}) 
    \label{eq:Z_lower}
\end{align}
hold, 
where the last inequality follows  \cref{lem:propset}, 
  which we will prove just below this proof, 
 with the fact that
  $\min_{v\in V}\aL_t(v)\geq \min_{v\in V}\aL_0(v) \geq \ave-1$ 
  since $\min_{v\in V}\aL_t(v)$ is monotone non-decreasing with respect to~$t$. 

 Then, we evaluate $\Pr[\sum_{t=0}^{T_2-1}Y'_t(a) <3 \log_2 \Delta ]$. 
 For convenience, 
 let $Z'_t$ ($t=0,1,\ldots,T_2-1$) be independent binary random variables 
   such that $\Pr[Z'_t=1]= \frac{c_*\fair}{3\rate}$,  and 
 let $Z'=\sum_{t=0}^{T_2-1}Z'_t$. 
 It is not difficult to see that 
   $\Pr[\sum_{t=0}^{T_2-1}Y'_t(a) <3 \log_2 \Delta ] \leq \Pr[\sum_{t=0}^{T_2-1}Z'_t(a) < 3\log_2 \Delta]$
   holds. 
Note that 
\begin{align}
 \E[Z'] 
  &= T_2 \Pr[Z'_t=1] \nonumber \\ 
  &\geq\frac{54\rate}{c_*\fair}\log(\Delta n / \varepsilon) \cdotp \frac{c_*\fair}{3\rate} 
    &&(\mbox{since $T_2 \geq \frac{54\rate}{c_*\fair}\log(\Delta n / \varepsilon)$})\nonumber\\ 
  & \geq 18\log(\Delta n/\varepsilon) \label{tmp20220524c} \\
  & \geq 12\log_2 \Delta &&(\mbox{since $n \geq 1$, $\varepsilon<1$ and $\log_2{\rm e} \leq 1.5$}) \label{tmp20220524d}
\end{align}
holds. 
Then, we have
\begin{align*}
    \Pr\left[\sum_{t=0}^{T_2-1}Z'_t(a) < 3\log_2 \Delta\right]
    &\leq \Pr\left[Z'\leq \frac{\E[Z']}{4}\right]&&(\mbox{by \eqref{tmp20220524d}})\nonumber\\
  &\leq \exp\left(-\frac{(\frac{3}{4})^2\E[Z']}{2}\right) &&(\mbox{by \cref{lem:Chernoff} (ii)})\nonumber\\
    &\leq \exp\left(-\frac{\E[Z']}{3}\right) \\
    &\leq \exp\left(-6\log(\Delta n/\varepsilon)\right) &&(\mbox{by \eqref{tmp20220524c}})\nonumber\\
  &\leq \left(\frac{\varepsilon}{\Delta n}\right)^6. 
\end{align*}

Finally, 
\begin{align*}
   \eqref{tmp20220608b}
   & \geq 1- \sum_{a \in A'} \left(\frac{\varepsilon}{\Delta n}\right)^6 \\
   & \geq 1- \Delta n\left(\frac{\varepsilon}{\Delta n}\right)^6 \\
   & = 1- \frac{\varepsilon^6}{(\Delta n)^5} \\
   & \geq 1- \varepsilon^2
\end{align*}
and we obtain (i). The proof of (ii) is similar. 
\end{proof}

\begin{lemma}
\label{lem:propset}
If $\aL \in \mathbb{N}^V$ 
  satisfies  $\min_{v\in V}\aL(v)\geq \ave-1$ 
 then $|S'(\aL)|\geq n/3$ 
where $S'(\aL) \defeq \{v\in V \mid \aL(v)\leq \rave\}$ and $\ave = \sum_{v \in V}\aL(v)/n$. 
\end{lemma}
\begin{proof}
Notice that 
\begin{align*}
    n \ave
    &=\sum_{v\in V}\aL(v) &&(\mbox{$\ave = \sum_{v\in V}\aL(v)/n$ by the definition})\\
    &=\sum_{v\in S'(\aL)}\aL(v)+\sum_{v\notin S'(\aL)}\aL(v)\\
    &\geq |S'(\aL)|(\ave-1)+\left(n-|S'(\aL)|\right)(\rave+1) &&(\mbox{by the definition of $S'(\aL)$})\\
    &\geq |S'(\aL)|(\ave-1)+\left(n-|S'(\aL)|\right)\left(\ave+\frac{1}{2}\right) &&(\mbox{since $\rave \geq \ave-1/2$})\\
    &=-\frac{3}{2}|S'(\aL)|+n\ave+\frac{n}{2}
\end{align*}
holds. Now, the claim is clear. 
\end{proof}

\subsection{Final step of the proof}\label{sec:final}
 By \cref{lem:PhaseI,lem:PhaseII}, we got a situation  
   $[\ave-1\leq \aL_{T_1+T_2}(v)\leq \rave+1] \vee [\bave-1\leq \bL_{T_1+T_2}(v) \leq \brave+1]$ hold for any $v \in V$ w.h.p. 
 The next lemma remarks that this is the situation which we want. 
\begin{lemma}\label{lem:logic}
Let $\aL \in \mathbb{N}^V$ where $\ave = \sum_{v \in V}\aL(v)$. 
For convenience,  
  let $\phi_1 \defeq [\ave-1\leq \aL(v)\leq \rave+1]$ and 
  let $\phi_2 \defeq [\bave-1\leq \bL(v) \leq \brave+1]$. 
Then, $\phi_1 \vee \phi_2$ implies $\aL(v)\in \{\rave-1,\rave,\rave+1\}$.
\end{lemma}
\begin{proof}
 Since $\bL(v)=\tL-\aL(v)$ and $\bave=\tL-\ave$ by their definitions, 
\begin{align}
 \phi_2
 &\Leftrightarrow 
   \tL-\ave-1\leq \tL - \aL(v) \leq \lceil \tL-\ave \rfloor +1 \nonumber\\
 &\Leftrightarrow
   \ave + 1\geq \aL(v) \geq - \lceil - \ave \rfloor - 1
\label{eq:psi2'}
\end{align}
holds. Thus, we see
\begin{align}
 &\phi_1 \vee \eqref{eq:psi2'} \nonumber\\
 &\Leftrightarrow 
  \min\{\ave,- \lceil -\ave \rfloor \}-1\leq \aL(v)\leq \max\{\ave, \rave\}+1
\label{tmp20220516c}
\end{align}
 holds.  
Concerning the most right-hand-side of \cref{tmp20220516c}, 
\begin{align*}
\max\{\ave, \rave\}+1 < \max\{\rave+1, \rave \}+1 \leq \rave + 2
\end{align*}
holds. Similarly, concerning the most left-hand-side of \cref{tmp20220516c}, 
\begin{align*}
\min\{\ave,- \lceil -\ave \rfloor\}-1 > \min\{\rave-1,\rave-1\}-1 \geq\rave-2
\end{align*}
holds.  Consequently, we see that 
\begin{align}
&\eqref{tmp20220516c} \nonumber\\
&\Rightarrow  \rave-2 < \aL(v) < \rave+2.
\label{tmp20220516d}
\end{align}
Since $\aL(v)$ is an integer, $\rave$ as well, 
   \cref{tmp20220516d} implies $\aL(v)\in \{\rave-1,\rave,\rave+1\}$. 
We obtain the claim. 
\end{proof}

Now, we finalize the proof of \cref{thm:RM_general}. 
\begin{proof}[Proof of \cref{thm:RM_general}]
Let $\Tbal = T_1 +T_2$ where $T_1$ and $T_2$ respectively satisfy the conditions in \cref{lem:PhaseI,lem:PhaseII}. 
Clearly, $\Tbal = \frac{90\rate}{c_*\fair}\log(\frac{\Delta n }{\varepsilon})+2$ is sufficient. 
Then, 
\begin{align*}
&\Pr[\aL_{\Tbal}(v)\in \{\rave-1,\rave,\rave+1\}] \\
 &\geq
  \Pr\left[\begin{aligned}
         &[\ave-1\leq \aL_{T_1+T_2}(v)\leq \rave+1]\\
         &\vee [\bave-1\leq \bL_{T_1+T_2}(v) \leq \brave+1]
    \end{aligned}
    \right] 
   \hspace{4em}(\mbox{by \cref{lem:logic}})
    \\
 &=
  \Pr\left[\begin{aligned}
         &[\ave-1\leq \aL_{T_1+T_2}(v)\leq \rave+1]\\
         &\vee [\bave-1\leq \bL_{T_1+T_2}(v) \leq \brave+1]
    \end{aligned}
    \,\middle|\,
  \begin{aligned}
         &[\min_{u\in V}\aL_{T_1}(u)\geq \ave-1]\\
         &\vee[\min_{u\in V}\bL_{T_1}(u)\geq \bave-1]
    \end{aligned} 
    \right]
 \cdotp
  \Pr\left[
  \begin{aligned}
         &[\min_{u\in V}\aL_{T_1}(u)\geq \ave-1]\\
         &\vee[\min_{u\in V}\bL_{T_1}(u)\geq \bave-1]
    \end{aligned} 
    \right]\\
 &=
  \Pr\left[\begin{aligned}
         &[\ave-1\leq \aL_{T_2}(v)\leq \rave+1]\\
         &\vee [\bave-1\leq \bL_{T_2}(v) \leq \brave+1]
    \end{aligned}
    \,\middle|\,
  \begin{aligned}
         &[\min_{u\in V}\aL_0(u)\geq \ave-1]\\
         &\vee[\min_{u\in V}\bL_0(u)\geq \bave-1]
    \end{aligned} 
    \right]
 \cdotp
  \Pr\left[
  \begin{aligned}
         &[\min_{u\in V}\aL_{T_1}(u)\geq \ave-1]\\
         &\vee[\min_{u\in V}\bL_{T_1}(u)\geq \bave-1]
    \end{aligned} 
    \right]\\
 &\geq (1-\varepsilon^2)^2
    \hspace{4em}(\mbox{by \cref{lem:PhaseIIb}, and \cref{lem:PhaseI} with \eqref{obs:PhaseIc}, respectively})\\
 &\geq 1-\varepsilon 
\end{align*}
   holds for any $v \in V$, where the last inequality follows the assumption $0<\epsilon<1/4$.  
 We obtain the claim. 
\end{proof}

\section{Implications of Theorem \ref{thm:RM_general}}
\label{sec:applications}
 The random matching (Algorithm~\ref{fig:algorithm})
   is a comprehensive method for load balancing on networks, seemingly quite natural and simple, 
  while there is variety how to draw a random matching, which is a profound issue in fact. 
 This section introduces three major varieties, 
  namely 
   {\em simple load balancing} (cf.,~\cite{BFKK19,HW22}), 
    {\em local random matching (LR) algorithm} (cf.,~\cite{GM96}), and 
    a simple variant of the {\em distributed synchronous algorithm} (cf.,~\cite{BGPS06,FS09,SS12,CS17}). 
 The main focus of this section is 
    to check their $\fair$-fair conditions \eqref{def:matching}, and 
   we show the upper bounds of their balancing times implied by \cref{thm:RM_general}. 
 

\paragraph{Supplemental terminologies about static graphs}
 To describe the algorithms, 
   we introduce supplemental terminologies about static graphs. 
 Suppose $G=(V,E)$ is an undirected simple graph. 
   Let $\Neb_E(v)$ (or simply $\Neb(v)$) denote the set of vertices adjacent to $v \in V$, 
   i.e., $\Neb(v) \defeq \{ u \in V \mid \{u,v\} \in E\}$.  
 We remark $v \not\in \Neb(v)$ since $G$ is simple. 
 We also remark $\deg(v) = |\Neb(v)|$ clearly. 
 Specially, we remark that we use $\Neb_{E'}(v)$ and $\deg_{E'}(v)$ for an edge subset $E' \subseteq E$, 
   i.e., $\Neb_{E'}(v) \defeq \{ u \in V \mid \{u,v\} \in {E'}\}$ and $\deg_{E'}(v) = |\Neb_{E'}(v)|$. 
 They are also abbreviated as $\Neb(v)$ and $\deg(v)$ on $E'$, without confusion. 

 In this section, we also deal with directed edges. 
 Let $\vec{G}=(V,\vec{E})$ denote a directed graph 
  where $\vec{E} \subseteq V^2$ denotes the set of directed edges. 
 Both $(u,v)$ and $(v,u)$ can simultaneously  exist in $\vec{E}$, but 
   duplication of a direct edge is not allowed, i.e., $(u,v)$ is at most one in $\vec{E}$. 
 A {\em self-loop} is NOT allowed here,  i.e., $(v,v)$ cannot exist in $\vec{E}$. 
 For a set of directed edges $\vec{E}$, 
  let $\Neb^-_{\vec{E}}(v)$ and  $\Neb^+_{\vec{E}}(v)$ (or simply $\Neb^-(v)$ and  $\Neb^+(v)$, resp.) 
   respectively denote the sets of in/out neighboring vertices of $v \in V$, 
   i.e., $\Neb^-(v) \defeq \{ u \in V \mid (u,v) \in \vec{E}\}$ and  
    $\Neb^+(v) \defeq \{ u \in V \mid (v,u) \in \vec{E}\}$. 
 Let 
   $\deg^-(v)$ and $\deg^-(v)$ respectively denote {\em in-degree} and {\em out-degree} on $\vec{E}$, 
  i.e.,    
   $\deg^-(v) \defeq |\Neb^-(v)|$ and $\deg^+(v) \defeq |\Neb^+(v)|$. 

\subsection{Simple load-balancing}\label{sec:simple}
\begin{algorithm}[t]
\caption{Random matching in simple load-balancing}
\label{alg:simple}
\SetKwInOut{Input}{Input}\SetKwInOut{Output}{Output}
\Input{$G=(V,E)$}
\Output{$M \subseteq E$ such that $M$ consists of a single edge, or is empty. }
\BlankLine
{set $M \defeq \emptyset$}\;
{choose $v \in V$ u.a.r.}\;
{\bf if} $\Neb(v) \neq \emptyset$ {\bf then} choose one $u \in \Neb(v)$ u.a.r., and put $\{u,v\} \in M$\;
\KwRet{$M$}\;
\end{algorithm}
 The simplest verity could be ``$M_t$ consists of a single edge.'' 
 Berenbrink et al. \cite{BFKK19} 
   spotlighted this folklore technology, and 
   they gave a simple analysis. 
 The algorithm is formally described in Algorithm~\ref{alg:simple}. 
 It is not difficult to see that 
\begin{align}
 \Pr[M = \{\{u,v\}\} ] 
  &\defeq \frac{1}{n}\frac{1}{\deg(u)}+ \frac{1}{n}\frac{1}{\deg(v)} \nonumber\\ 
  &= \frac{1}{n} \left(\frac{1}{\deg(u)}+\frac{1}{\deg(v)}\right) \nonumber\\
  &\geq \frac{1}{n} \left(\frac{1}{\max\{\deg(u),\deg(v)\}}\right) 
\label{prob:simple}
\end{align}  
 holds for any $\{u,v\} \in E$, 
   meaning that Algorithm~\ref{alg:simple} satisfies $\fair=1/n$-fair condition~\eqref{def:matching}. 
 Thus, \cref{thm:RM_general,prob:simple} imply that 
   the algorithm achieves $\Tbal =\Order(\rate n \log(n\Delta))$ with high probability on edge-Markovian graphs. 

\paragraph{A variant: Draw from $E$ u.a.r. } 
 To avoid the situation that  Algorithm~\ref{alg:simple} output $M_t=\emptyset$,  
  a reader may think 
   why not choose an edge of $E$ uniformly at random. 
 In this variant, 
\begin{align}
 \Pr[M = \{\{u,v\}\}] 
  & \defeq \frac{1}{|E|} \nonumber\\ 
  & \geq \frac{1}{n^2}\left(\frac{1}{\max\{\deg(u),\deg(v)\}}\right) 
\label{prob:simple2}
\end{align} 
 holds,  
  where the last inequality is tight in an order of magnitude, 
    in a cerebrated lollipop graph for an instance, or in the graph $K_{n-2}+K_2$ for another instance. 
 \cref{thm:RM_general,prob:simple2} imply that the algorithm achieves 
   $\Tbal =\Order(\rate n^2 \log(n\Delta))$ with high probability on edge-Markovian graphs. 
 Of course, this upper bound might not be tight.

\subsection{LR algorithm}
\begin{algorithm}[t]
\caption{Random matching in local randomized (LR) algorithm}
\label{alg:LR}
\SetKwInOut{Input}{Input}\SetKwInOut{Output}{Output}
\Input{$G=(V,E)$}
\Output{$M \subseteq E$ is a matching}
\BlankLine
{set $\vec{M} \defeq \emptyset$, set $\vec{R}^- \defeq\emptyset$,  set $\vec{R}^+ \defeq\emptyset$\;}
{{\bf for} $v \in V$, {\bf for} $u \in N(v)$, put $(v,u) \in \vec{M}$ w.p. $\frac{1}{8\max\{\deg(v),\deg(u)\}}$\;}
{{\bf for} $v \in V$, {\bf if}  $\deg^+(v) >1$ on $\vec{M}$ 
  {\bf then} put $(v,u) \in \vec{R}^+$ for all $u \in \Neb^+(v)$ on $\vec{M}$\;}
{set $\vec{M}:=\vec{M} \setminus \vec{R}^+$\;}
{{\bf for} $v \in V$, {\bf if} $\deg^-(v) >1$ on $\vec{M}$
  {\bf then} put $(u,v) \in \vec{R}^-$ for all $u \in \Neb^-(v)$ on $\vec{M}$\;}
{set $\vec{M}:=\vec{M} \setminus \vec{R}^-$\;}
 \KwRet{$M \defeq \left\{ \{v,u\} \,\middle|\, [(v,u) \in \vec{M}]\wedge [\mbox{$\deg^-(v) = 0$ on $\vec{M}$}] \right\} 
  \cup \left\{ \{v,u\}  \,\middle|\,  [(v,u) \in \vec{M} ]\wedge  [(u,v) \in \vec{M}]\right\}$}\;
\end{algorithm}
 To draw a random matching, 
  it could be a natural idea 
   to choose a random subset of edges and then to edit it to a matching. 
 Ghosh and Muthukrishnan in~\cite{GM96}
    originally proposed the random matching scheme for load balancing (cf., Algorithm~\ref{fig:algorithm}), 
    referred to as the {\em local randomized} ({\em LR}) algorithm, 
  where they gave such an algorithm to generate a random matching as well. 
 The algorithm is summarized in Algorithm~\ref{alg:LR}, and 
   it could be regarded in line with the above natural idea. 
 For Algorithm~\ref{alg:LR}, 
  Ghosh and Muthukrishnan~\cite{GM96} proved the following lemma. 
\begin{lemma}[\cite{GM96}]\label{lemma:LR}
$\Pr[\{u,v\} \in M] \geq \frac{1}{8\max\{\deg(v),\deg(u)\}}$ holds for any $\{u,v\} \in E$.
\end{lemma} 
  \cref{thm:RM_general,lemma:LR} imply that the algorithm achieves 
   $\Tbal = \Order (\rate \log(n\Delta))$ with high probability on edge-Markovian graphs. 
 

%
\subsection{(A localized version of) distributed synchronous algorithm}
\label{sec:distributed}
\begin{algorithm}[t]
\caption{Random matching in a variant of distributed synchronous algorithm}
\label{alg:ds}
\SetKwInOut{Input}{Input}\SetKwInOut{Output}{Output}
\Input{$G=(V,E)$}
\Output{$M \subseteq E$ is a matching}
\BlankLine
{set $\vec{M} = \emptyset$, set $I=\emptyset$\;}
{{\bf for} $v \in V$, put $v \in I$ w.p. $\frac{1}{2}$\;}
\For{ $v \in I$ such that $\Neb(v) \neq \emptyset$}{ 
   choose one $u \in \Neb(v)$ u.a.r.\;  
   put $(v,u) \in \vec{M}$ w.p. $\frac{\min\{\deg_G(u),\deg_G(v)\}}{\deg_G(u)}$\tcp*[r]{Metropolis---Hastings (ll. 4---5)}}
 \KwRet{$M \defeq \left\{ \{v,u\} \,\middle|\, [(v,u) \in \vec{M}]\footnotemark \wedge [u \not\in I] \wedge [\deg^-(u)=1 \mbox{ on $\vec{M}$}]   \right\}$}\;
\end{algorithm}
\footnotetext{Formally, this could be $[v \in I] \wedge [(v,u) \in \vec{M}]$. 
  However, we remark $[(v,u) \in M']$ implies $[v \in I]$ by the algorithm, meaning that  $[v \in I]$ is redundant. }
 Another natural idea for a random matching 
  may be to let vertices choose random partners. 
 This idea works pretty well on bipartite graphs, 
   while it needs attention for non-bipartite graphs. 
 The {\em distributed synchronous algorithm}, 
    given by Boyd et al.~\cite{BGPS06} in a bit different context,
   could be regarded in line with the above idea. 
 Some works \cite{FS09,SS12,CS17} about  the random matching algorithm for load balancing 
  employs this algorithm and analyzed it.

 Here, we describe a variant of the algorithm in Algorithm~\ref{alg:ds}, 
   where we localize the choice of a partner (ll. 4--5) based on the Metropolis-Hastings technique (cf., \cite{LP17}) 
   instead of using the globally maximum degree in the original algorithm~\cite{BGPS06}.  
 The lines 4 and 5 of Algorithm~\ref{alg:ds} 
   realize the probability
\begin{align}
  \Pr[(u,v) \in \vec{M} \mid v \in I]
  &= \begin{cases}
  \frac{1}{\deg(v)}\frac{\deg(u)}{\deg(u)} = \frac{1}{\deg(v)} & (\mbox{if $\deg(u) \leq \deg(v)$}) \\
  \frac{1}{\deg(v)}\frac{\deg(v)}{\deg(u)} = \frac{1}{\deg(u)} & (\mbox{if $\deg(u) > \deg(v)$}) 
\end{cases} \nonumber \\
& =\frac{1}{\max \{\deg(u),\deg(v)\}}
\end{align}
for each $u \in \Neb(v)$, which provides the following lemma. 

\begin{lemma}\label{lem:ds}
Algorithm~\ref{alg:ds} satisfies $\fair=1/4$-fair condition, i.e., 
\begin{align*}
  \Pr[\{u,v\} \in M] \geq \frac{1}{4 \max \{\deg(u),\deg(v)\}}
\end{align*}
   holds for any $\{u,v\} \in E$. 
\end{lemma}

\begin{proof}
By the line 7 of Algorithm~\ref{alg:ds}, 
\begin{align}
\Pr[\{v,u\} \in M] 
 &= 
   \Pr[v \in I] \cdotp \Pr[(v,u) \in \vec{M} \mid v \in I] \cdotp 
   \Pr[u \not\in I] \cdotp \prod_{v' \in \Neb(u) \setminus \{v\} } \Pr[(v',u) \not\in \Vec{M}] \nonumber \\
 &\hspace{1em} + 
  \Pr[u \in I] \cdotp \Pr[(u,v) \in \vec{M} \mid u \in I] \cdotp 
  \Pr[v \not\in I] \cdotp \prod_{u' \in \Neb(v) \setminus\{u\}}\Pr[(u',v) \not\in \Vec{M}]  \nonumber\\
 &= \frac{1}{2} \cdotp \frac{1}{\max \{\deg(u),\deg(v)\}} \cdotp \frac{1}{2} \cdotp 
  \prod_{v' \in \Neb(u) \setminus \{v\} } \Pr[(v',u) \not\in \Vec{M}]  \nonumber\\
 &\hspace{1em}
   + \frac{1}{2} \cdotp \frac{1}{\max \{\deg(u),\deg(v)\}} \cdotp \frac{1}{2}  \cdotp 
  \prod_{u' \in \Neb(v) \setminus \{u\} } \Pr[(u',v) \not\in \Vec{M}]
\label{tmp20220628a}
\end{align}
holds. 
Here, 
\begin{align}
\prod_{w \in \Neb(v) \setminus \{v\} } \Pr[(w,v) \not\in \Vec{M}]
 &= \prod_{w \in \Neb(v) \setminus \{v\} } \left(1-\Pr[w \in I] \Pr[(w,v) \in \Vec{M} \mid w \in I]\right) \nonumber\\
 &= \prod_{w \in \Neb(v) \setminus \{v\} }  \left(1-\frac{1}{2 \max \{\deg(w),\deg(v)\}} \right)  \nonumber\\
 &\geq 1- \sum_{w \in \Neb(v) \setminus \{v\} }  \frac{1}{2 \max \{\deg(w),\deg(v)\}} &&(\mbox{union bound})  \nonumber \\
 &\geq 1- \sum_{w \in \Neb(v) \setminus \{v\} }  \frac{1}{2 \deg(v)} \nonumber \\
 &\geq 1-\frac{\deg(v)}{2\deg(v)}  \nonumber\\
 &= \frac{1}{2}
\label{tmp20220628b}
\end{align}
holds. By \cref{tmp20220628a,tmp20220628b}, we obtain the claim. 
\end{proof}

  \cref{thm:RM_general,lem:ds} imply that the algorithm achieves 
   $\Tbal = \Order (\rate \log(n\Delta))$ with high probability on edge-Markovian graphs.

\section{Concluding Remark}
 Motivated by a technique for an analysis of algorithms on dynamic graphs, 
   this paper gave an upper bound of the balancing time of random matching algorithms 
     for load balancing on {\em edge-Markovian} graphs, 
   which is a major topic in the context of distributed computing. 
 To avoid the difficulty caused by the complicated correlation in the history of executions, 
  we have developed a technique of {\em history-independent} bound 
   \cref{lem:prob_RM,ineq:prob_RM,eq:Z_lower2}, 
  focusing on the $\fair$-fair factor which existing algorithms have.

 Concerning our bound for the load-balancing algorithms, 
  the $\rate = \max\left\{\frac{1-q}{p},\frac{p}{1-q}\right\}$ factor in \cref{thm:RM_general}
   could be improved by more careful arguments. 
 Concerning the random matching algorithm, 
  an extension of the analysis technique to random edge-subset algorithm is an interesting future work. 
 Concerning the edge-Markovian graph, 
  an extension to more general model, particularly vertex increasing model, is an important future work. 






\section*{Acknowledgements}
This work was partially supported by the joint project of Kyoto University and Toyota Motor Corporation, titled ``Advanced Mathematical Science for Mobility Society''.

\bibliographystyle{plain}
\bibliography{ref}
\appendix
\section{Tools}
\begin{lemma}[Lemma 1.8.7 in \cite{Doerr18}]
\label{lem:domination}
Let $X_1,\ldots,X_n$ be arbitrary binary random variables and $X_1^*,\ldots,X_n^*$ be independent binary random variables. Let $X=\sum_{i=1}^nX_i$ and $X^*=\sum_{i=1}^nX_i^*$.
Suppose that $\Pr[X_i=1\mid X_1=x_1,\ldots,X_{i-1}=x_{i-1}]\geq \Pr[X_i^*=1]$ for all $1\leq i\leq n$ and all $x_1,\ldots,x_{i-1}\in \{0,1\}$ with $\Pr[X_1=x_1,\ldots,X_{i-1}=x_{i-1}]>0$. Then, 
$\Pr[X\leq \lambda]\leq \Pr[X^*\leq \lambda]$ for all $\lambda\in \mathbb{R}$.
\end{lemma}
\begin{lemma}[The Chernoff inequality (see, e.g., Theorem 1.10.21 in \cite{Doerr18})]
\label{lem:Chernoff}
Let $X_1,\ldots,X_n$ be $n$ independent random variables taking values in $[0,1]$. Let $X=\sum_{i=1}^nX_i$. 
Let $\mu^{-}\leq \E[X]\leq \mu^+$.
Then, we have the following:
\begin{enumerate}[label=(\roman*)]
\item $\Pr\left[X\geq (1+\epsilon)\mu^+\right]\leq \exp\left(-\frac{\min\{\epsilon^2,\epsilon\}\mu^+}{3}\right)$ for $\epsilon\geq 0$.
\item $\Pr\left[X\leq (1-\epsilon)\mu^-\right]\leq \exp\left(-\frac{\epsilon^2\mu^-}{2}\right)$ for $0\leq \epsilon\leq 1$.
\end{enumerate}
\end{lemma}

\section{Convergence behavior of edge-Markovian graphs}\label{sec:conv}
For a technical reason, we assume $p\in(0,1)$ and $q\in(0,1)$ here\footnote{
  The case of $p \in \{0,1\}$ or $q \in \{0,1\}$ is similar, 
   but it needs some treatment in individual cases. 
  }. 
\begin{proposition}
 $\Pr[\{u,v\} \in E_t]$ approaches $p/(p+q)$ assymptotic to $t$. 
\end{proposition}
\begin{proof}
Clearly $P$ is ergodic, meaning that the unique limit distribution is the stationary distribution. 
Thus,  just check $(\frac{q}{p+q},\frac{p}{p+q})P=(\frac{q}{p+q},\frac{p}{p+q})$ holds. 
\end{proof}
\begin{proof}[another proof]
Note that 
\begin{align*}
 P - (1-p-q)I &= \begin{pmatrix} 1-p -(1-p-q) & p \\ q & 1-q -(1-p-q)\end{pmatrix}\\
   &= \begin{pmatrix} q & p \\ q & p\end{pmatrix}
\end{align*}
holds, where $I$ is the identity matrix. In other words, 
\begin{align}
 P &= (1-p-q)I + (p+q)\begin{pmatrix} \frac{q}{p+q} & \frac{p}{p+q} \\ \frac{q}{p+q} & \frac{p}{p+q}\end{pmatrix}
\label{eq:P1}
\end{align}
 holds. Thus, 
\begin{align*}
 \begin{pmatrix} \frac{q}{p+q} & \frac{p}{p+q} \end{pmatrix} P 
 &= 
  \begin{pmatrix} \frac{q}{p+q} & \frac{p}{p+q} \end{pmatrix}
  \left((1-p-q)I +  (p+q)\begin{pmatrix} \frac{q}{p+q} & \frac{p}{p+q} \\ \frac{q}{p+q} & \frac{p}{p+q}\end{pmatrix}\right) \\
 &= 
 (1-p-q)\begin{pmatrix} \frac{q}{p+q} & \frac{p}{p+q} \end{pmatrix} 
  + (p+q)\begin{pmatrix} \frac{q}{p+q} & \frac{p}{p+q} \end{pmatrix}  \\
 &= 
  \begin{pmatrix} \frac{q}{p+q} & \frac{p}{p+q} \end{pmatrix} 
\end{align*}
  and we obtain the claim. 
\end{proof}

\begin{proposition} 
 Let $\boldsymbol{\pi}=\left(\frac{q}{p+q},\frac{p}{p+q}\right)$, for convenience. 
 If $p+q=1$ then $\frac{1}{2}\|\boldsymbol{x}P^t-\boldsymbol{\pi}\|_1 =0$ 
  for any $t=0,1,2,\ldots$ and  
  for any probability distribution $\boldsymbol{x}=(x_1,x_2)$, 
     i.e., $x_1 \geq 0$, $x_2 \geq 0$ and $x_1+x_2=1$. 
 Otherwise, 
  for any $\epsilon$ ($0<\epsilon<1$), 
   $\frac{1}{2}\|\boldsymbol{x}P^t-\boldsymbol{\pi}\|_1 \leq \epsilon$ 
  holds 
   for $t \geq \log \epsilon/\log|1-p-q|$ and 
   for any probability distribution $\boldsymbol{x}$. 
\end{proposition}
\begin{proof}
When $p+q=1$, the claim is clear by \cref{eq:P1}. 
Suppose $p+q \neq 1$. 
Let $t \geq \log \epsilon/\log|1-p-q|$, then 
 \cref{eq:P1} implies 
\begin{align*}
 \frac{1}{2} \left\|\boldsymbol{x}P^t-\boldsymbol{\pi}\right\|_1 
 &= |1-p-q|^t \frac{ \left|x_1-\frac{q}{p+q}\right| + \left|x_2-\frac{p}{p+q}\right| }{2} \\
 &\leq |1-p-q|^t \\
 &\leq |1-p-q|^{\frac{\log \epsilon}{\log|1-p-q|}} \\
 &= \epsilon
\end{align*}  
 holds for any distribution $\boldsymbol{x}$.
\end{proof}

\end{document}